\documentclass[journal]{IEEEtran}
\usepackage{tabularx} % extra features for tabular environment
\usepackage{amsmath,amsfonts,amssymb,mathtools,mathrsfs,mathabx}
\usepackage{amsthm}
\usepackage{enumerate}
\usepackage{indentfirst}
\usepackage{colortbl}
\usepackage{multirow}
\usepackage{color}
\usepackage{graphicx} % takes care of graphic including machinery
\usepackage{subfigure}
\usepackage{cite} % takes care of citations
\usepackage{hyperref} % adds hyper links inside the generated pdf file
\usepackage{cleveref} 
\usepackage{lineno}
\usepackage{algpseudocode,algorithm,algorithmicx}
\usepackage{tikz}
\usetikzlibrary{plotmarks}

\newcommand{\commentout}[1]{}
\newcommand{\CC}{\mathbb{C}}
\newcommand{\RR}{\mathbb{R}}
\newcommand{\EE}{\mathbb{E}}

\newcommand{\be}{\boldsymbol{e}}
\newcommand{\bx}{\boldsymbol{x}}

\newcommand{\hY}{\widehat{Y}}

\newcommand{\hU}{\widehat{U}}

\newcommand{\calR}{\mathcal{R}}
\newcommand{\calJ}{\mathcal{J}}

\newcommand{\calE}{\mathcal{E}}
\newcommand{\calH}{\mathcal{H}}

\newcommand{\md}{{\rm md}}

\renewcommand{\subset}{\subseteq}
\renewcommand{\hat}{\widehat}

\renewcommand{\epsilon}{\varepsilon}
\renewcommand{\Re}{\text{Re}}

\def\SRF{{\rm SRF}}

\def\<{\big\langle}
\def\>{\big\rangle}
\def\({\Big(}
\def\){\Big)}

\def\C{\mathbb{C}}

\def\e{\boldsymbol{e}}
\def\E{\mathbb{E}}

\def\calH{\mathcal{H}}

\def\P{\mathbb{P}}

\def\R{\mathbb{R}}
\def\calR{\mathcal{R}}

\def\x{\boldsymbol{x}}

\def\y{\boldsymbol{y}}
\def\T{\mathbb{T}}

\def\Z{\mathbb{Z}}

\DeclareMathOperator{\dist}{dist}
\DeclareMathOperator{\diag}{diag}
\DeclareMathOperator{\rank}{rank}

\newtheorem{theorem}{Theorem}[section]

\newtheorem{lemma}[theorem]{Lemma}

\theoremstyle{definition}

\newtheorem{definition}[theorem]{Definition}

\newtheorem{remark}[theorem]{Remark}
\newtheorem{assumption}[theorem]{Assumption}

\numberwithin{equation}{section}

\begin{document}

\title{Stability and Super-resolution of MUSIC and ESPRIT for Multi-snapshot Spectral Estimation}

\author{Weilin~Li, 
        ~Zengying~Zhu, ~Weiguo~Gao,        and~Wenjing~Liao% <-this % stops a space
\thanks{W. Li is with the City College of New York, e-mail: wli6@ccny.cuny.edu. 
This research is supported by the AMS-Simons Travel Grant.
}
\thanks{Z. Zhu is with School of Mathematical Sciences, Fudan University, Shanghai, China, e-mail: zengyingzhu@fudan.edu.cn.}
\thanks{W. Gao is with School of Mathematical Sciences, Fudan University, Shanghai, China, e-mail: wggao@fudan.edu.cn. 
This research is partially supported by National Key R$\&$D Program of China under Grant No. 2021YFA1003305.
}
\thanks{W. Liao is with School of Mathematics, Georgia Institute of Technology, Atlanta, GA, USA, e-mail: wliao60@gatech.edu. This research is supported by NSF DMS 1818751, DMS 2012652 and NSF CAREER DMS 2145167.
}
}

\maketitle

\begin{abstract}
This paper studies the spectral estimation problem of estimating the locations of a fixed number of point sources given multiple snapshots of Fourier measurements collected by a uniform array of sensors. We prove novel stability bounds for MUSIC and ESPRIT as a function of the noise standard deviation, number of snapshots, source amplitudes, and support. Our most general result is a perturbation bound of the signal space in terms of the minimum singular value of Fourier matrices. When the point sources are located in several separated clumps, we provide an explicit upper bound of the noise-space correlation perturbation error in MUSIC and the support error in ESPRIT in terms of a Super-Resolution Factor (SRF). The upper bound for ESPRIT is then compared with a new Cram\'er-Rao lower bound for the clumps model. As a result, we show that ESPRIT is comparable to that of the optimal unbiased estimator(s) in terms of the dependence on noise, number of snapshots and SRF. As a byproduct of our analysis, we discover several fundamental differences between the single-snapshot and multi-snapshot problems. Our theory is validated by numerical experiments. 
\end{abstract}

\IEEEpeerreviewmaketitle

\section{Introduction}

\subsection{Problem formulation and motivation}

This paper studies the spectral estimation problem of estimating the locations of a fixed number of point sources given multiple time snapshots of Fourier measurements collected by a uniform array of sensors. Let $S$ be the number of point sources which we assume are located in the set $\Omega:=\{\omega_k\}_{k=1}^S\subseteq \mathbb{T}$, where $\mathbb{T}:=[0,1)$ is the torus.  We denote the source amplitudes at time $t>0$ by the complex-valued vector $\boldsymbol{x}(t):=\{x_j(t)\}_{j=1}^S\in \mathbb{C}^S$. At time $t>0$, a uniform array of $M\geq S$ sensors collects a noisy measurement vector 
	\begin{equation}
		\y(t)
		:=\Phi \x(t)+\e(t) \in \CC^{M},
		\label{eqlinearsystem}
	\end{equation}
	where $\Phi:=\Phi_M(\Omega) \in \mathbb{C}^{M \times S}$ is the Fourier sensing matrix with entries
	\[
	\Phi_{k,j}
	=e^{-2\pi i k\omega_j}, 
	\]
	for $ k=0,\dots,M-1$ and $j=1,\dots,S,$
	and $\be(t) \in \C^{M}$ represents a noise vector at time $t$. 
	The sensing matrix $\Phi_M(\Omega)$ only depends on $\Omega$ and $M$, and is independent of time $t$. It does not necessarily contain orthogonal columns unless $\Omega$ is a subset of $\{k/M\}_{k=0}^{M-1}$.
	
	{Typically in practice, samples of $\y(t)$ are collected at $L$ times $t_1<\cdots<t_L$ and each $\y(t_\ell)$ is called a {\it snapshot}. This paper studies the {\it multi-snapshot spectral estimation} problem of estimating the source locations $\Omega$ from these $L$ snapshots: $\{\y(t_\ell)\}_{\ell =1}^L$. } 
	This problem appears in many interesting imaging and signal processing applications, including 
	inverse scattering \cite{fannjiang2011music}, Direction-Of-Arrival (DOA) estimation \cite{krim1996array,schmidt1986multiple,ottersten1993exact} and spectral analysis \cite{stoica1997introduction}. 

	Various methods have been developed by the imaging and signal processing communities for multi-snapshot spectral estimation \cite{prony1795essai,schmidt1986multiple,kailath1989esprit,hua1990matrixpencil,barabell1983improving}. {A class of algorithms commonly referred to as subspace methods are widely used in applications due to their impressive empirical performance, especially in the context of super-resolution imaging where some point sources are close together}. 
	In particular, MUSIC \cite{schmidt1986multiple} and ESPRIT \cite{kailath1989esprit} are among the most popular subspace methods. {MUSIC and ESPRIT are applicable in multi-dimensions as well \cite{roemer2014analytical}.} {The success of MUSIC and ESPRIT has been demonstrated in many simulations and applications \cite{tschudin1999comparison}. }
	
	Despite great developments in numerical methods {and extensive empirical studies}, there are many open theoretical questions regarding the stability and super-resolution of subspace methods. This paper focuses on the fundamental {performance analysis} question: \textit{What is the stability of MUSIC and ESPRIT, as a function of the noise standard deviation, the number of snapshots, source amplitudes, and the support set?}

	Prior theoretical work on related super-resolution problems \cite{donoho1992superresolution,li2020super,li2021stable} suggest that the stability of the multi-snapshot spectral estimation problem crucially depends on the locations of the point sources. The {\it minimum separation} of $\Omega$ is %defined as
\[
\Delta
:=\Delta(\Omega)
:=\min_{1\leq j<k\leq S} |\omega_j-\omega_k|_\T,
\]
where $|\omega|_\T:=\min_{n\in\Z} |\omega-n|$. Since $M$ consecutive Fourier coefficients are collected during each time snapshot, the standard imaging resolution is $1/(M-1)$, which is referred to as the {\it Rayleigh Length (RL)} in optics \cite{den1997resolution}.

We expect the stability of this problem to be different depending on the relationship between $\Delta$ and $1/(M-1)$. In this article, we are particularly interested in is the {\it super-resolution regime} where $\Delta < 1/(M-1)$, so that there exist $\omega_j$ and $\omega_k$ whose distance is smaller than 1 RL. When this is the case, we define the {\it super-resolution factor} (SRF) by 
\begin{equation}
\SRF := \frac{1}{(M-1)\Delta}
\geq 1. 
\label{eqsrf}
 \end{equation}

	\subsection{Main results}
	
	The main contribution of this paper is a detailed performance analysis of MUSIC and ESPRIT in terms of the fundamental model quantities: the number of snapshots, noise standard deviation, support set, and source amplitudes.
		\begin{enumerate}[(a)]
	
	\item 
	Our most general upper bounds for the perturbation of MUSIC (Theorem \ref{thmmusicstability}) and ESPRIT (Theorem \ref{thmespritstability}) show that their average errors can be controlled by a term that is linear in 
	\begin{equation}
	\frac{1}{{\sigma}_S(\Phi)\sqrt{\lambda_S(X)}} \times \frac{\text{Noise}}{\sqrt L} ,
	\label{eqpropbound}
	\end{equation}
	where Noise represents the noise standard deviation to be specified in Assumption \eqref{assump:main}, ${\sigma}_S(\Phi)$ denotes the $S$-th largest singular value of $\Phi$, and $\lambda_S(X)$ is the minimum eigenvalue of the amplitude covariance matrix $X$ to be defined in \eqref{eqX}. This inequality not only captures the correct dependence on the number of snapshots and noise, but it also highlights how the stability of both algorithms implicitly depends on the configuration of the support set $\Omega$ through the quantity $\sigma_S(\Phi)$.
	
	\item
	To give a more transparent upper bound in the super-resolution regime, we consider a specific scenario where $\Omega$ consists of separated clumps and the point sources in each clump can be close together. Our upper bounds for the stability of MUSIC (Theorem \ref{thm:musicsuper1}) and ESPRIT (Theorem \ref{thm:espritsuper1}) are proportional to 
	$$\frac{{\SRF}^{\lambda-1} }{\sqrt{{\lambda}_S(X)}} \times \frac{\text{Noise}}{\sqrt L},$$
	where $\lambda$ is the cardinality of the largest clump. This indicates that, for challenging super-resolution problems where $\SRF^{\lambda-1}$ is large, additional snapshots or higher quality samples must be taken for compensation.  
	
	\item 
	We prove a new Cram\'er-Rao lower bound (Theorem \ref{thm:CR}) under a specific separated clumps model in Theorem \ref{thm:CR}. This lower bound matches our upper bound for ESPRIT in terms of the dependence on noise, $L$ and $\SRF$, thereby certifying that the performance of ESPRIT is comparable to that of the optimal unbiased estimator(s) for this model.

	\end{enumerate}

	\subsection{Comparison and connection to other works}
	
	Prior resolution analysis of subspace methods for multi-snapshot spectral estimation \cite{zhang1995probability,ferreol2010statistical,rao1989performance}  focused on a special case where there are only two closely spaced point sources. The papers \cite{zhang1995probability,ferreol2010statistical} analyzed the probability that two sources are correctly detected instead of being misspecified as a single source. It was shown in  \cite{rao1989performance} that the ESPRIT support error is upper bounded by a term on the order of ${\SRF}/({\sqrt{M L}} \times  \text{SNR)}$ for certain Signal-to-Noise Ratio (SNR) defined in the referenced article.
			 
	To our knowledge, there are no other theoretical works that address more complicated situations beyond two closely spaced point sources, e.g., when the support contains multiple clumps of point sources and the point sources in each clump can be closely spaced. A theoretical analysis for more complicated situations could be valuable. For example, a recent article \cite{liu2020joint} empirically compared the performance of several numerical methods, including MUSIC, for DOA estimation of several point sources arranged in complicated ways. {Other models have been considered, see \cite{de2004target, xu1997new}, but this direction is beyond the scope of this work.}

	In literature, many existing works have addressed  the sensitivity of multi-snapshot MUSIC \cite{friedlander1990sensitivity,swindlehurst1992performance,li1992sensitivity,friedlander1994effects} and ESPRIT \cite{swindlehurst1990sensitivity,li1992sensitivity,nekrutkin2010perturbation}.
Many of these works focus on sensitivity to model errors, which arise from antenna
array perturbations, sensor gain and phase errors, diagonal
noise covariance perturbations, etc.	
	{A first-order perturbation analysis is given in \cite{friedlander1990sensitivity,swindlehurst1992performance,li1992sensitivity,li1993performance,friedlander1994effects,rao1989performance} for MUSIC, in \cite{swindlehurst1990sensitivity,li1992sensitivity,nekrutkin2010perturbation} for ESPRIT, in \cite{roemer2014analytical} for tensor-ESPRIT in the multidimensional case, and in \cite{djermoune2009perturbation} for the estimation of a damped complex exponential.
	 The results in \cite{li1993performance,rao1989performance,roemer2014analytical} and many others  prove that, the sensitivity of MUSIC and ESPRIT is proportional to Noise$/\sqrt L$. 	 However, they are
	  implicit for a super-resolution analysis since the dependence on SRF is often hidden in some matrix eigenvalue.    }
	  
	  {When MUSIC and ESPRIT are used in applications, many interesting techniques have been developed to improve its performance. Spatial smoothing \cite{pillai1989performance,rao1990effect,steinwandt2017performance} is widely used when some sources are coherent, or if only a small number of snapshots is available. A first-order perturbation analysis with spatial smoothing is given in \cite{pillai1989performance,rao1990effect} and in \cite{steinwandt2017performance} in the multi-dimensional case. When the source amplitudes are non-circular, some analysis can be found in \cite{steinwandt2017performance}. Mathematical theories on the sensitivity of MUSIC and ESPRIT usually assume that the number of sources $S$ is known. In practice, an accurate estimation of $S$ is an interesting problem. 
	  We refer to \cite{zhao1986detection,wax1989detection,kritchman2009non,liavas2001behavior,xu1994detection} for some interesting techniques and statistical analysis on the estimation of $S$.
	  }

	The classical Cram\'er-Rao Bound (CRB) \cite{cramer1999mathematical} provided a lower bound on the variance of any unbiased estimator of the support, 	and was extensively investigated in  \cite{stoica1989music,stoica1990music,lee1992cramer,koochakzadeh2016cramer}. The paper that is most relevant to us is \cite{lee1992cramer}, which gave a Cram\'er-Rao lower bound when all $S$ point sources are almost equally spaced in a small interval. This result is summarized in detail in Section \ref{secCRBA}.
	
	In recent years, spectral estimation has been extensively studied in the single-snapshot scenario where $L=1$, and primarily from a deterministic viewpoint. 
	Performance guarantees of a convex optimization algorithm were established in \cite{candes2013super, 
	duval2015exact} for sufficiently separated point sources and in \cite{morgenshtern2016super,benedetto2020super} for super-resolution. On the other hand, MUSIC and ESPRIT are non-convex methods, and their stability and super-resolution limits were addressed in \cite{li2021stable,li2019conditioning,li2020super}. A detailed review for the single-snapshot setting can be found in \cite{li2021stable}. There are fundamental differences between the single-snapshot and multi-snapshot problems, which we discuss in Section \ref{sec:singlevsmultiple}. Due to these differences, the results in this paper can not be directly deduced from the ones in \cite{li2021stable,li2019conditioning,li2020super}. 
	
	A key quantity in the mathematics of super-resolution is $\sigma_S(\Phi)$, which crucially depends on the support geometry. Explicit lower bounds for $\sigma_S(\Phi)$ were derived in \cite{liao2016music,moitra2015matrixpencil,batenkov2020conditioning,kunis2020condition,diederichs2019well,kunis2020smallest,batenkov2020spectral,li2021stable} for various support models. This paper uses the lower bound under the separated clumps model in \cite{li2021stable}.

	\subsection{Organization}
	     This paper is organized as follows. We introduce assumptions and define important terminology in Section \ref{secassumption}. We define the signal space associated with this problem and derive an error between the true and empirical signal spaces in Section \ref{secsignal}. We present our stability analysis for MUSIC and ESPRIT in Section \ref{secstability}, and then consider a separated clumps model in Section \ref{secsuperresolution}. To understand the fundamental limits of spectral estimation, we deduce a new Cram\'er-Rao lower bound under a specific separated clumps model in Section \ref{seccramerrao} and compare it to our upper bound for ESPRIT. We present numerical experiments in Section \ref{section:num_exp} and discuss the differences between the multi-snapshot versus single-snapshot problems in Section \ref{sec:singlevsmultiple}. All proofs are contained in Section \ref{sec:proofs}.

\subsection{Notation}
We use $\RR$ and $\CC$ to denote the set of real and complex numbers respectively.
	For $z \in \CC$, we denote its angle by $\arg z \in [0,2\pi)$. For a vector $\x \in \CC^{S} $, we denote  the diagonal matrix with $\x$ on the diagonal by $\diag(\bx) \in \CC^{S \times S}$. We let $\|\x\|$ denote the Euclidean norm of $\x$ {and $\x^*$ be its conjugate transpose. Let $I_k$ be the $k\times k$ identity matrix and $\delta_{t,s}$ be the Kronecker delta; that is, $\delta_{t,s}= 1$ if $t=s$, and $\delta_{t,s}= 0$ if $t\neq s$. }

	For a matrix $A$, we use $\sigma_1(A)\geq \sigma_2(A)\geq \cdots$ to denote the singular values of $A$ listed in non-increasing order, and each singular value appears according to its multiplicity. If $A$ has real eigenvalues, we use $\lambda_1(A)\geq \lambda_2(A)\geq \cdots$ to denote its eigenvalues of $A$. {We use $\|A\|_2$ and $\|A\|_F$ for the spectral and Frobenius norms of $A$, respectively. Let $A^*$ be the conjugate transpose of $A$.} We use $\odot$ to denote the Hadamard (pointwise) product of two matrices of equal size.
	For a linear operator $A$, we denote its range by $R(A)$. For any subspace $U$, we let $P_U$ be the orthogonal projection onto $U$. Slightly abusing notation, we let $P_A$ be the orthogonal projection onto $R(A)$. 
	
	We use $Z\sim \mathcal{G}_{\nu,\tau}$ to specify that $Z\in \R$ is a random variable such that 
	$\E Z=0$, $\text{var}(Z)=\nu^2$, and $\E \exp(u Z)\leq \exp(\tau \nu^2  u^2)$ for all $u\in\R$. For a complex $Z$, we write $Z\sim \mathcal{G}_{\nu,\tau}$ to mean that the real and imaginary parts are independent $\mathcal{G}_{\nu/\sqrt 2,\tau}$ random variables. If $Q$ is an event, then $Q^c$ denotes its complement and $1_Q$ is the indicator function of $Q$.

    {We use the notation $A\lesssim_{m,n,p} B$ to mean that there exists a $C>0$ depending only on $m,n,p$ such that $A\leq CB$.}

Throughout the paper, we assume that the noise vectors $\be(t)$ are random, and the expectation $\E$ is taken over the probability distribution of the noise $\be(t)$. 

	\section{Assumptions and covariance matrices}
	\label{secassumption}

	 Our performance guarantees for MUSIC and ESPRIT require the following standard assumptions.

	\begin{assumption}
		\label{assump:main}
		(Model assumptions) Fix positive integers $L,M,S$ such that $S\leq M$ and $S\leq L$.
		\begin{enumerate}[(a)]
			\item
			$\Omega\subset\T$ has cardinality $S$ and does not depend on $t$. 
		    \item
			The amplitude covariance matrix, 
			\begin{equation}
			\label{eqX}
			X:=\frac{1}{L}\sum_{\ell=1}^L \x(t_\ell)\x(t_\ell)^*,
			\end{equation}
			is strictly positive-definite. 
			\item
			There exist $\nu,\tau>0$ such that for any $t>0$, the entries of $\e(t)\in\C^M$ are independent complex $\mathcal{G}_{\nu,\tau}$ random variables. For $t\not=s$, assume $\e(t)$ and $\e(s)$ are independent. These assumptions imply that
			\begin{align*}
			\E \e(t)\e(s)^* = \nu^2 \delta_{t,s} I_{M}. 
			\end{align*}
			
			\item
			For each $\ell = 1,\ldots,L$, we are given 
			\[
			\y(t_\ell) := \Phi \x(t_\ell) + \e(t_\ell).
			\]
			\item
			The noise level satisfies the following assumption
			\begin{equation}
				\label{eq:nu}
				\nu\leq \sigma_S(\Phi)\sqrt{\lambda_S(X)}.
			\end{equation}
			
		\end{enumerate}
	\end{assumption}
	
	Assumption \ref{assump:main} is standard in spectral and DOA estimations \cite{schmidt1986multiple,krim1996array,stoica1989music}. This assumption can be justified as follows:

	Assumption \ref{assump:main}(a) is necessary from a statistical viewpoint. If $\Omega$ is allowed to change over time, then there might be no relationship between $\y(t_k)$ and $\y(t_\ell)$ for $k\not=\ell$, so collecting additional snapshots might not be beneficial. Under this assumption, $\Phi:=\Phi(\Omega)\in\C^{M\times S}$ is a special type of matrix, often referred to as a non-harmonic Fourier matrix or Vandermonde matrix with nodes on the complex unit circle, and by the celebrated Vandermonde determinant theorem, $\Phi$ has rank $S$ and $\sigma_S(\Phi)>0$.

	Assumption \ref{assump:main}(b) on the strictly positive-definiteness of the amplitude covariance matrix is standard in array imaging and DOA estimation \cite{schmidt1986multiple,krim1996array,stoica1989music}. This assumption is equivalent to the requirement that $\x(t_1),\dots,\x(t_L)\in\C^S$ span $\C^S$. If this assumption is void, both MUSIC and ESPRIT can be modified as the single-snapshot scenario by utilizing a Hankel structure. This is a technical point to be discussed in Section \ref{sec:singlevsmultiple}. 
	
	For example, Assumption \ref{assump:main}(b) can be used in models where $\x(t)$ evolves over time according to a physical law. Many stochastic models of $\x(t)$ fulfill Assumption \ref{assump:main}(b). For instance, if $\x(t_1),\dots,\x(t_L)$ are independent Gaussian vectors with $N(0,\nu^2)$ i.i.d. entries, then $X$ has full rank with probability one. More generally, if $\x(t)$ is sub-Gaussian and the population covariance of $\x(t)$ is strictly positive-definite, then $X$ has full rank with high probability, {see  \cite{vershynin2018high}[Theorem 4.71 and Exercise 4.7.3].}

	When $L$ snapshots of measurements are taken, we define the matrices 
		\begin{equation}
		\begin{aligned}
		\label{eqXEYL}
		X_L &:=
		\frac{1}{\sqrt L}\begin{bmatrix}
			\x(t_1) &\cdots &\x(t_L) 
		\end{bmatrix}, \\
		E_L &:=
		\frac{1}{\sqrt L}\begin{bmatrix}
			\e(t_1) &\cdots &\e(t_L)
		\end{bmatrix}, \\
		Y_L &:=
		\frac{1}{\sqrt L}\begin{bmatrix}
			\y(t_1) &\cdots &\y(t_L) 
		\end{bmatrix}.
		\end{aligned}
		\end{equation}
		The empirical covariance matrices are
		\begin{align}
		X 
		&:=\frac{1}{L} \sum_{\ell=1}^L \x(t_\ell) \x(t_\ell)^*
		=X_L X_L^* \in \CC^{S \times S}, 
		\label{eq:hatX}
		\\
		\hat Y
		&:=\frac{1}{L} \sum_{\ell=1}^L \y(t_\ell) \y(t_\ell)^*
		=Y_L Y_L^*\in \CC^{M \times M}. 
		\label{eq:hatY}
	\end{align}
	Defining $Y:=\Phi X\Phi^*$, one can verify that  
	\begin{equation}
			\mathbb{E}\hat Y 
			= Y+\nu^2 I_M.
			\label{eqhatY}
		\end{equation}
		
\section{Signal  space and its empirical estimation}
\label{secsignal}

\subsection{Signal and noise space}
A signal space plays an important role in a class of subspace methods \cite{krim1996array}, including MUSIC and ESPRIT. These algorithms first compute an empirical version of the signal space by a truncated eigen-decomposition, and then extract the source locations based on this signal space. 
	
\begin{definition}[Signal and noise spaces]
	The {\it signal space} is defined to be the column space of $\Phi$, and the {\it noise space} is defined to be the orthogonal complement of the signal space.
\end{definition}
	
Although the signal and noise spaces do not depend on the specific choice of orthonormal basis, we pick particular bases purely for convenience. Throughout this paper, we let $U \in \CC^{M \times S}$ and $W \in \CC^{M \times (M-S)}$ such that $U^* U =I $ and $ W^* W= I$ be matrices whose columns form an orthonormal basis for the signal space and the noise space, respectively. Due to identity \eqref{eqhatY} and Assumption \ref{assump:main}(b), the signal space is the eigenspace associated with the $S$ largest eigenvalues of $Y:=\Phi X\Phi^*$. 
	 
With $L$ snapshots of \textit{noisy} measurements, the signal and noise spaces can be estimated from the empirical covariance matrix $\hat Y$ as follows. Let $[\hat U,\hat W]\in\C^{M\times M}$ be a unitary matrix whose columns are eigenvectors of $\hat Y$, and the columns of $\hat U \in \CC^{M \times S}$ are the eigenvectors corresponding to the $S$ largest eigenvalues of $\hat Y$.

\begin{definition}
    We refer to $\hat U$ and $\hat W$ as the {\it empirical signal space} and the {\it empirical noise space}, respectively.
\end{definition}
     
\subsection{An accurate perturbation bound for the signal space}  
    
	In this paper, we establish an accurate perturbation bound for $\hat U$. Roughly speaking, if the number of snapshots is sufficiently large, and everything else is fixed, the empirical signal space $\hat U$ is a good approximation of the signal space $U$. We first quantify the distance between $\hat U$ and $U$ by their canonical angles.
	\begin{definition}
	Suppose $U , \hat U \in \CC^{M \times S}$ and $U^* U = \hat U^* \hat U = I$. The canonical angles between $U$ and $\hat U$ are defined to be
	$\theta_j = \arccos\ \sigma_j(\hat U^* U)$, for $j=1,\dots,S$. Let
	\[
	\sin\theta(\hat U, U)
	:=\diag\big( \sin(\theta_1), \cdots, \sin(\theta_S) \big). 
	\]
	The (Euclidean) $\sin\theta$ distance between $U$ and $\hat U$ is
	\begin{equation}
	\dist(\hat U,U):=\|\sin\theta(\hat U,U)\|_2.
	\end{equation}
	\end{definition} 
	This definition extends the notion of an angle between two vectors to subspaces, and it is invariant to the particular choice of basis. It follows from {\cite[Chapter I, Theorem 5.5]{Stewart90}  and \cite[Lemma 2.1.2]{chen2020spectral}} that
	\begin{equation}
	\begin{aligned}
	\dist(\hat U,U) &= \|(I-P_{\hat U})P_U\|_2 \\
	&=\|(I-P_U)P_{\hat U}\|_2
	=  \|P_U-P_{\hat U}\|_2.
	\label{eqdistuhatuproj}
	\end{aligned}
	\end{equation}
 	
	{It is a consequence of \cite[Theorem 3]{zhang1995probability} that the $\sin\theta$ distance between $U$ and $\hat U$ satisfies the following expectation bound. See Section \ref{secproofthm:Uperturb} for the details.
	\begin{theorem}
		\label{thm:Uperturb}
		Under Assumption \ref{assump:main}, there exists a constant $C_\tau>0$ depending only on $\tau$ such that
		\[
		\E \big( \dist(\hat U,U)^2 \big)
		\leq \frac{C_\tau M}{ \lambda_S(X) \sigma_S^2(\Phi)} \frac{\nu^2}{L}. 
		\]
	\end{theorem}
	\begin{remark}
		Note that $X$ defined in \eqref{eqX} is normalized by $L$, so we can interpret $\lambda_S(X)$ has a normalized quantity. Since $\Phi$ has $S$ columns each with Euclidean norm $\sqrt M$, we can also interpret $\sqrt M/\sigma_S(\Phi)$ as a normalized quantity. Hence, Theorem \ref{thm:Uperturb} tells us that the average squared $\sin\theta$ distance between $U$ and $\hat U$ tends to zero proportional (at least) linearly in $\nu^2/L$. 
	\end{remark}
	\begin{remark}
		Throughout this paper, we will provide upper bounds for the average squared $\sin\theta$ distance. We can use Jensen's inequality, 
		$
		\big(\E \dist(\hat U,U) \big)^2
		\leq \E\big( \dist(\hat U,U)^2 \big),
		$
		to convert them into an upper bound for the expected error if desired.  
	\end{remark}
	}

	%%%%%%%%%%%%%%%%%%%%%%%%%%%%%%%%%%%%%%%%%%%
	\section{Stability of MUSIC and ESPRIT}
	\label{secstability}

	\subsection{Review of MUSIC and ESPRIT}
	
	In signal processing, MUSIC \cite{schmidt1986multiple} and ESPRIT \cite{kailath1989esprit} have been widely used in applications due to their impressive numerical performance and super-resolution phenomenon. 

We first introduce some important notations.
	We define the signal vector at $\omega\in [0,1)$ as
	\[
	\phi(\omega) 
	:=\begin{bmatrix} 1 &e^{2\pi i \omega} &e^{4\pi i \omega} &\cdots &e^{2\pi i (M-1) \omega} \end{bmatrix}^* \in \CC^{M}.
	\]
	The columns of $\Phi$ in \eqref{eqlinearsystem} exactly consist of the signal vectors at the sources locations:
	\[
	\Phi 
	= \begin{bmatrix}
		\phi(\omega_1) &\phi(\omega_2) &\cdots &\phi(\omega_S)
	\end{bmatrix}.
	\]

		In the noiseless situation, MUSIC amounts to finding the noise space, forming the noise-space correlation function and identifying the zero set of the noise-space correlation, which is necessarily the support set.
When $M \ge S+1$, we have the following observations regarding the Vandermonde matrix $\Phi$:
	\begin{itemize}
		\item $\rank(\Phi) = S$ 
		\item $\operatorname{rank}\left(\left[\Phi,\phi(\omega)\right]\right)=S+1$ if and only if $\omega \notin \Omega$.
	\end{itemize}
	These observations imply that, $\omega \in \Omega$ if and only if $\phi(\omega) \in R(\Phi)$. Since the signal space $U$ is exactly $R(\Phi)$, the condition of $\phi(\omega) \in R(\Phi)$ can be quantified by  a noise-space correlation function or an imaging function.

	\begin{definition}[Noise-space correlation (NSC) and imaging functions]
	For any $\omega \in [0,1)$,
		the noise-space correlation function is defined as
		\[
		\calR(\omega)
		:= \frac{\|(I-P_{U}) \phi(\omega)\| }{\|\phi(\omega)\|}
		= \frac{\big\|(I-P_{U}) \phi(\omega) \big\|}{\sqrt M}.
		\]
		%where $P_U$ is the orthogonal projection to the column space of $U$ (which by definition is equal to the column space of $\Phi$).
		The imaging function is defined as 
		$$\calJ(\omega) := {1}/{\calR(\omega)}.$$
		\end{definition}
	The MUSIC algorithm is based on the following lemma:	
	\begin{lemma}
	\label{lemmamusicnoisefree}
		Let $M \ge S+1$. Then
		$$
		\omega \in \Omega \Longleftrightarrow \mathcal{R}(\omega)=0  \Longleftrightarrow \calJ(\omega) = \infty.
		$$	
	\end{lemma}
	Lemma \ref{lemmamusicnoisefree} implies that, the source locations can be exactly identified through the roots of the noise-space correlation function, or the peaks of the imaging function. 

		When $L$ snapshots of noisy measurements are taken, we can compute the empirical signal space $\hat U$.  
		
		\begin{definition}[Empirical NSC and imaging functions]
		For any $\omega\in[0,1)$, the empirical noise-space correlation (NSC) function is defined to be 
		$$\widehat{\mathcal{R}}(\omega) := \frac{\| (I-P_{\hat U})\phi(\omega)\|}{\| \phi(\omega) \|}.$$
	The empirical imaging function is defined to be $$\widehat{\mathcal{J}}(\omega)  := {1}/{\widehat{\mathcal{R}}(\omega)}.$$	
		\end{definition}
		An estimate for the support set is obtained by extracting the $S$ largest local maxima of $\widehat{\mathcal{J}}$, or equivalently, the $S$ smallest local minima of $\widehat{\mathcal{R}}$. MUSIC is summarized
		 in Algorithm \ref{algmusicesprit}.

	\begin{algorithm}[t]
		\begin{algorithmic}
			\Require Measurements $\{\boldsymbol{y}(t_\ell)\}_{\ell=1}^{L}$ and sparsity $S$.
			
			\State 1. Form the empirical covariance matrix $\hY$ according to equation \eqref{eq:hatY}.
			\State 2. Compute the eigen-decomposition of $\widehat{Y}$ to obtain $\hU$, the empirical signal space.
			\State {\bf MUSIC:}  Compute  the empirical imaging function
			$$\widehat{\mathcal{J}}(\omega) = \frac{\| \phi(\omega) \|}{\| (I-P_{\hat U})\phi(\omega)\|},\ \omega \in [0,1),$$
			and identify its $S$ largest local maxima  as $\left\{\widehat{\omega}_{j}\right\}_{j=1}^{S}$.
			\vspace{0.05cm}
			\State {\bf ESPRIT:} Let $\widehat{U}_{0}$ and $\widehat{U}_{1}$ be two submatrices of $\widehat{U}$ containing the first and the last $M-1$ rows respectively. Compute
			$$
			\widehat{\Psi}=\widehat{U}_{0}^{\dagger} \widehat{U}_{1}
			$$
			and its $S$ eigenvalues $\widehat{\lambda}_{1}, \ldots, \widehat{\lambda}_{S}$. Set $\widehat{\omega}_{j}:=-\frac{\arg\widehat{\lambda}_{j}}{2 \pi}$.
			\Ensure $\widehat{\Omega}=\left\{\widehat{\omega}_{j}\right\}_{j=1}^{S}$. 			
					\end{algorithmic}
		\caption{MUSIC and ESPRIT}
		\label{algmusicesprit}
	\end{algorithm}

	ESPRIT reformulates the support estimation step as an eigenvalue problem. We first discuss the noiseless situation. By definition, the signal space $U$ is the column space of $\Phi$, and so there exists an invertible matrix  $Q \in \CC^{S\times S}$ such that $U = \Phi Q.$ Let $U_0$ and $U_1$ be two submatrices of $U$ containing the first and the last $M-1$ rows respectively. Denote $\Phi_{M-1} \in \CC^{(M-1)\times S} $ as the submatrix containing the first $M-1$ rows of $\Phi$. Letting $D_{\Omega} := \diag(e^{-2\pi i \omega_1},\dots e^{-2\pi i\omega_S}) \in \CC^{S \times S}$, we have
	\[
		U_0 = \Phi_{M-1} Q, \quad 
		U_1 = \Phi_{M-1} D_{\Omega} Q.
	\]
	Setting $M-1 \ge S$ guarantees that $U_0$ and $U_1$ have rank $S$. It follows that
	\begin{equation*}
		\Psi := U_0^{\dagger} U_1 = Q^{-1} D_{\Omega} Q\in\C^{S\times S}. 
	\end{equation*}
	Hence, the eigenvalues of $\Psi$ are exactly $\{e^{-2\pi i \omega_j}\}_{j=1}^S$. The ESPRIT algorithm amounts to finding the support set $\Omega$ through the eigenvalues of $\Psi$.

	When $L$ snapshots of noisy measurements are taken,  ESPRIT follows the same steps but with the empirical signal space $\hat U$ instead of $U$. It is summarized in Algorithm \ref{algmusicesprit}.
	
	To quantify the error between $\Omega$ and $\hat \Omega$, we introduce the support matching distance. 
	
	\begin{definition}
		The matching distance between $\Omega$ and $\widehat{\Omega}$ is %defined as
		$$
		\operatorname{md}(\Omega, \widehat{\Omega}):= \min _{\psi} \max _{j}\left|\widehat{\omega}_{\psi(j)}-\omega_{j}\right|_{\mathbb{T}}
		$$
		where $\psi$ is a permutation on $\{1,\dots,S\}$.
	\end{definition}

	\subsection{Stability of MUSIC and ESPRIT}
	\label{secmusicespritstability}
	 The stability of MUSIC depends on the perturbation of the NSC function from $\calR$ to $\hat \calR$, which can be measured as % by the $L^\infty$ norm as
	\[
	\|\widehat\calR - \calR\|_{\infty} := \sup_{\omega \in [0,1)} |\widehat\calR(\omega) - \calR(\omega)|.
	\]

	{
	\begin{theorem}
		\label{thmmusicstability}
		Let $M \ge S+1$.
		Under Assumption \ref{assump:main} and with $C_\tau>0$ being the constant in Theorem \ref{thm:Uperturb}, we have
		\begin{equation}
			\E \big( \|\hat\calR -\calR\|_\infty^2 \big)
			\leq \frac{C_\tau M}{ \lambda_S(X) \sigma_S^2(\Phi)} \frac{\nu^2}{L}.
		\label{thmmusicstabilityeq}
		\end{equation}
	\end{theorem}

	Theorem \ref{thmmusicstability} is proved in Section \ref{proofthmmusicstability}.

 	Unlike MUSIC, ESPRIT is an explicit algorithm since it reformulates the main question as an eigenvalue problem. For ESPRIT, we have the following estimate for the matching distance between the true and recovered support.
	 	\begin{theorem}
		\label{thmespritstability}
		Let $M \ge S+1$. Under Assumption \ref{assump:main}, let $C_\tau>0$ be the constant in Theorem \ref{thm:Uperturb} and $\hat\Omega$ be the output of ESPRIT.  
		\begin{enumerate}[(a)]
			\item 
			\underline {Moderate SNR regime.} We have
			\begin{equation*}
				\E \big( {\rm md}(\hat\Omega,\Omega)^2 \big) \leq  \frac{C_\tau 16^{S+2}  S^{3} M^2}{\lambda_S(X)\sigma_{S}^4(\Phi)} \frac{\nu^2}{L}.
			\end{equation*}			
			\item 
			\underline{Large SNR regime.} Define 
    	    \begin{align*}
    	    \xi:=\frac{\sigma_S^2(\Phi) \lambda_S(X)}{ M} \frac{L}{\nu^2}, \quad 
    	    \rho:=\frac{1}{\sqrt 6} \frac{4^{S+2}\sigma_{S}^2(\Phi)\Delta}{S^2 M}.
    	    \end{align*}
    	    There is a sufficiently large $D_\tau\geq 1$ depending only on $\tau$ such that if
        	\begin{equation}
        		\label{eq:SNRbound}
        		 \xi \geq D_\tau \max\Big(1,\frac{1}{\rho},\frac{1}{\rho^2}\Big),
        	\end{equation}
        	then it holds that 
        	$$
        	\E\big( \md(\hat \Omega,\Omega)^2 \big) 
        	\lesssim_\tau \frac{M}{\sigma_S^2(\Phi) \lambda_S(X)} \frac{\nu^2}{L}. 
        	$$
		\end{enumerate}		
	\end{theorem}
}

    {Due to the considerable length of the proof of Theorem \ref{thmespritstability}, we split the proof into two parts. Parts (a) and (b) are proved in Sections \ref{proofthmespritstabilitya} and \ref{proofthmespritstabilityb} respectively. Part (a) holds for all possible values of $L$ and $\nu$. On the other hand, part (b) requires that $\xi$ (which can be interpreted as a scaled and squared signal-to-noise-ratio) is sufficiently large depending on the model parameters.}

	\section{Super-resolution of multi-snapshot MUSIC and ESPRIT}
	\label{secsuperresolution}

	\subsection{Minimum singular value of Fourier matrices }
	
	In this section, we consider a general model of $\Omega$
 where the point sources are clustered into  separated clumps.
	\begin{assumption}\label{separatedclumpsmodel}
		(Separated clumps model). Let $M$ and $R$ be a positive integers and $\Omega \subseteq \mathbb{T}$ have cardinality $S .$ We say that $\Omega$ consists of $R$ separated clumps with parameters $(R,M, S, \alpha, \beta)$ if the following hold.	
		\begin{enumerate}[(a)]
		    \item 
		    $\Omega$ can be written as the union of $R$ disjoint sets $\left\{\Lambda_{r}\right\}_{r=1}^{R},$ where each $\operatorname{clump} \Lambda_{r}$ is contained in
	    	an interval of length $1 /(M-1)$. 
		    \item 
		    $\Delta \geq \alpha / (M-1)$ with $\max _{1 \leq r\leq R}\left(\lambda_{r}-1\right)<1 / \alpha$
		    where $\lambda_{r}$ is the cardinality of $\Lambda_{r}$.
		
		    \item
		    If $R>1,$ then the distance between any two clumps is at least $\beta / (M-1)$.
		\end{enumerate}
	\end{assumption}

\begin{figure}[h]
	\centering
	\begin{tikzpicture}[xscale = 0.6,yscale = 0.6]
	\draw[thick] (-6,0) -- (-0.5,0);
	\filldraw[red] (-5,0) circle (0.1cm);		
	\filldraw[red] (-4.7,0) circle (0.1cm);		
	\filldraw[red] (-4.4,0) circle (0.1cm);		
	\draw[blue,thick,<->] (-4.7,-0.2) -- (-4.4,-0.2);
	\node[blue,below] at (-4.4,-0.3) {$\frac{\alpha}{M-1}$};
	
	%\draw[blue,thick,<->] (-4.4,0.2) -- (-2,0.2);
	%\node[blue,above] at (-3.2,0.2) {$\beta/M$};
	
	\node[above] at (-4.7,0.2) {$\Lambda_1$};
	\filldraw[red] (-2,0) circle (0.1cm);		
	\filldraw[red] (-1.7,0) circle (0.1cm);		
	\filldraw[red] (-1.4,0) circle (0.1cm);		
	\draw[blue,thick,<->] (-1.7,-0.2) -- (-1.4,-0.2);
	\node[blue,below] at (-1.4,-0.3) {$\frac{\alpha}{M-1}$};
	\node[above] at (-1.7,0.2) {$\Lambda_2$};
	\draw[dotted,thick] (-0.5,0) -- (1,0);
	\draw[thick] (1,0) -- (6.4,0);
	\filldraw[red] (2,0) circle (0.1cm);		
	\filldraw[red] (2.3,0) circle (0.1cm);		
	\filldraw[red] (2.6,0) circle (0.1cm);		
	\draw[blue,thick,<->] (2.3,-0.2) -- (2.6,-0.2);
	\node[blue,below] at (2.6,-0.3) {$\frac{\alpha}{M-1}$};
	\node[above] at (2.3,0.2) {$\Lambda_{R-1}$};
	\filldraw[red] (5,0) circle (0.1cm);		
	\filldraw[red] (5.3,0) circle (0.1cm);		
	\filldraw[red] (5.6,0) circle (0.1cm);		
	\draw[blue,thick,<->] (5.3,-0.2) -- (5.6,-0.2);
	\node[blue,below] at (5.6,-0.3) {$\frac{\alpha}{M-1}$};
	\node[above] at (5.3,0.2) {$\Lambda_R$};
	\draw[blue,thick,<->] (2.6,0.2) -- (5,0.2);
	\node[blue,above] at (3.8,0.2) {$\frac{\beta}{M-1}$};
	%%  
	%\node[below] at (-0.5,-1) {$\lambda =|\Lambda_1| = |\Lambda_2|= \ldots = |\Lambda_A| $};
	\end{tikzpicture}
	\caption{$\Omega = \cup_r \Lambda_r$ where each $\Lambda_r$ contains 3 equally spaced points with spacing $\alpha/(M-1)$. The clumps are separated at least by $\beta/(M-1)$. %The distance between clumps is $\beta/M$. 
	}
	\label{FigDemoClumps1}
\end{figure}
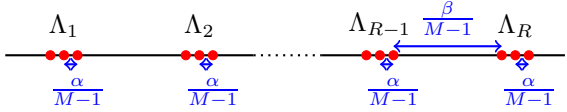
	
	An example of separated clumps is given in Figure \ref{FigDemoClumps1}.
	There are many types of discrete sets that consist of separated clumps. Extreme examples include when $\Omega$ is a single clump containing all $S$ points, and when $\Omega$ consists of $S$ clumps containing a single point. While our theory applies to both extremes, the in-between case where $\Omega$ consists of several clumps each of modest size is most interesting. A super-resolution theory of \textit{single-snapshot} MUSIC and ESPRIT for this  separated clumps model is developed in \cite{li2021stable,li2020super,li2019conditioning}. This paper focuses on the multi-snapshot scenario. The difference between the single-snapshot and multi-snapshot cases is discussed in Section \ref{sec:singlevsmultiple}. 

Under this separated clumps model, $\sigma_{S}(\Phi)$ can be estimated as an $\ell^2$ aggregate of $R$ terms, where each term only depends on the ``geometry" of each clump \cite[Theorem 2.7]{li2021stable}.

	\begin{theorem}
	\label{lemmasingular}
		Let $M \geq S^{2}+1$. Assume $\Omega$ satisfies Assumption \ref{separatedclumpsmodel} with parameters $(R,M, S, \alpha, \beta)$ for some $\alpha>0$ and
		\begin{equation}
		\label{betaeq}
		\beta \geq \max _{1 \leq r \leq R} \frac{20 S^{1 / 2} \lambda_{r}^{5 / 2}}{\alpha^{1 / 2}}.
		\end{equation}
		Then there exist explicit constants $B_1,\dots,B_R>0$ where $B_r$ only depends on $M$ and $\lambda_r$ such that 
		$$
		\sigma_{S }(\Phi) \geq \sqrt{M-1}\Big(\sum_{r=1}^{R}\big(B_{r} \alpha^{-\lambda_{r}+1}\big)^{2}\Big)^{-\frac{1}{2}}. 
		$$
				\end{theorem}
		
		An explicit formula for $B_r$ is given in  \cite[Theorem 2.7, Eq. (2.5)]{li2021stable}. In particular, $B_r$ only depends on $\lambda_r$ and $M$ (although $B_r$ can be further upper bounded in terms of only $\lambda_r$ if desired) and importantly, $B_r$ does not depend on $\alpha$.
	
The main feature of this theorem is the exponent on $\alpha$, which depends on the cardinality of each clump as opposed to the total number of points $S$. Let $\lambda$ be the cardinality of the largest clump: \begin{equation}
\label{eqlambdamax}
\lambda = \max_{1\leq r\leq R} \lambda_r.
\end{equation}
Since the inequality holds for any $\Omega$ that is a $(R,M,S,\alpha,\beta)$ set, then it holds for $\alpha=\Delta (M-1)$.
{Defining the super-resolution factor $\SRF$ as equation \eqref{eqsrf}, Lemma \ref{lemmasingular} implies
\begin{equation}
\label{eqlowert1}
\sigma_{\min}(\Phi) \ge
C \sqrt{M-1}\ \SRF^{-\lambda+1}.
\end{equation}

\subsection{Super-resolution of MUSIC and ESPRIT}
In this section, we combine our stability analysis for multi-snapshot MUSIC and ESPRIT in Section \ref{secmusicespritstability} with the minimum singular value estimate in Theorem \ref{lemmasingular} to derive a super-resolution theory for multi-snapshot MUSIC and ESPRIT.

For MUSIC, we obtain an upper bound for the perturbation of the noise-space correlation function by combining Theorems \ref{thmmusicstability} and \ref{lemmasingular}. 

\begin{theorem}
		\label{thm:musicsuper1}
		Suppose $M \ge S+1$, $\Omega$ satisfies the separated clumps model in Assumption \ref{separatedclumpsmodel} with parameters $(R,M,S,\alpha,\beta)$ {such that \eqref{betaeq} holds. Under Assumption \ref{assump:main}, let  $C_\tau$ and $B_r$ be the constants in Theorems \ref{thm:Uperturb} and \ref{lemmasingular} respectively, and let $\hat\calR$ be the perturbed noise-space correlation function of MUSIC. Then we have 
		\begin{equation}
		\E \big( \|\hat\calR -\calR\|_\infty^2 \big) \leq \frac{C_\tau }{\lambda_S(X)} \sum_{r=1}^{R}\left(B_{r} \alpha^{-\lambda_{r}+1}\right)^{2}  \frac{\nu^2}{L}.
		\label{thm:musicsuper1eq1}
		\end{equation}
		}
	\end{theorem}

Theorem \ref{thm:musicsuper1} provides a perturbation bound for the noise-space correlation function in multi-snapshot MUSIC under the separated clumps model. The terms that appear in the upper bound in Theorem \ref{thm:musicsuper1} can be interpreted as follows: {$\nu^2$ is the noise variance, $1/L$ is the usual stochastic factor, and the quantity in front of $\nu^2/L$ can be interpreted as the condition number of super-resolution recovery of multi-snapshot MUSIC.} Letting $\lambda$ be the cardinality of the largest clump defined in \eqref{eqlambdamax}, {Theorem \ref{thm:musicsuper1} shows that 
    \begin{equation}
    \E \big( \|\hat\calR -\calR\|_\infty^2 \big)
    \lesssim_{M,S,\lambda,\tau} \frac{ \SRF^{2\lambda -2}}{\lambda_S(X)} \frac{\nu^2}{L}.
    \label{eqmusicscaling}
    \end{equation}}
    Notice that the right hand side only depends on the cardinality of the largest clump {instead of the total number of point sources $S$}. This perturbation bound on MUSIC is verified by numerical experiments in Section \ref{section:num_exp}.

As for ESPRIT, we provide a bound for the support error by combining Theorem \ref{thmespritstability} and  Theorem \ref{lemmasingular}.

\begin{theorem}
		\label{thm:espritsuper1}
		Suppose $M \ge S+1$, $\Omega$ satisfies the separated clumps model in Assumption \ref{separatedclumpsmodel} with parameters $(R,M,S,\alpha,\beta)$ such that {\eqref{betaeq} holds. Under Assumption \ref{assump:main}, let $C_\tau$ and $B_r$ be the constants in Theorems \ref{thm:Uperturb} and \ref{lemmasingular} respectively, and let $\hat\Omega$ be the output of ESPRIT.
		\begin{enumerate}[(a)]
			\item 
			\underline{Moderate SNR regime.} We have
			\begin{equation*}
				\E \big( {\rm md}(\hat\Omega,\Omega)^2\big) 
				\leq \frac{C_\tau 16^{S+2}  S^{3}}{\lambda_S(X)} \Big(\sum_{r=1}^{R} \left(B_{r} \alpha^{-\lambda_{r}+1}\right)^{2}\Big)^2 \, \frac{\nu^2}{L}.
			\end{equation*}
			\item 
			\underline{Large SNR regime.} Define 
    	    \begin{align*}
    	    \rho
    	    :=\frac{1}{2\sqrt 6} \frac{4^{S+2}\Delta}{S^2 M}\sum_{r=1}^{R} \left(B_{r} \alpha^{-\lambda_{r}+1}\right)^{2}.
    	    \end{align*}
    	    There is a sufficiently large $D_\tau\geq 1$ depending only on $\tau$ such that if
        	\begin{equation*}
        		 \frac{L}{\nu^2} \geq \frac{D_\tau}{\lambda_S} \sum_{r=1}^{R}\big(B_{r} \alpha^{-\lambda_{r}+1}\big)^{2}  \max\Big(1,\frac{1}{\rho},\frac{1}{\rho^2}\Big),
        	\end{equation*}
        	then it holds that 
        	$$
        	\E\big( \md(\hat \Omega,\Omega)^2 \big) 
        	\lesssim_\tau \frac{1}{\lambda_S(X)}  \sum_{r=1}^{R}\big(B_{r} \alpha^{-\lambda_{r}+1}\big)^{2} \,  \frac{\nu^2}{L}. 
        	$$
		\end{enumerate}
	}

	\end{theorem}

Theorem \ref{thm:espritsuper1} provides an estimate of the support error in multi-snapshot ESPRIT under the separated clumps model. Let $\lambda$ be the cardinality of the largest clump defined in \eqref{eqlambdamax}. 
{If the assumptions in part (b) hold, then we can provide a simpler estimate, 
\begin{equation}
\E\big( \md(\hat \Omega,\Omega)^2 \big)
\lesssim_{M,S,\lambda,\tau} \frac{ \SRF^{2\lambda -2}}{\lambda_S(X)}\frac{\nu^2}{L}. 
\label{eqespritscaling}
\end{equation}
}
This perturbation bound on ESPRIT is verified by numerical experiments in Section \ref{section:num_exp}.

{
In comparison, classical perturbation bounds on MUSIC and ESPRIT usually depend on some eigenvalues and eigenvectors of the noiseless covariance matrix $Y$, such as (29a, 29b) of \cite{rao1989performance}. If we know the true source locations and amplitudes, we can compute the eigenvalues and eigenvectors, and then plug them into (29a) of \cite{rao1989performance}. In real applications, we do not know the true source locations and amplitudes. Therefore, those bounds can not be directly evaluated.

In the special case of two sources, (33a) of \cite{rao1989performance} shows that the squared
source localization error is proportional to $(1/\Delta)^2 = {\SRF}^2$ ($\Delta$ is referred as the one defined in \cite{rao1989performance}). The two-source scenario is a special case of our clumps model where
the support consists of a single clump with two point sources: $R = 1$ and $\lambda= 2$.
Therefore, (33a) of \cite{rao1989performance} is a special case of our Theorem \ref{thm:musicsuper1} and Theorem \ref{thm:espritsuper1}. Our results
consider a more general class of support, and such generalization is highly nontrivial.
}

	\section{Cram\'er-Rao lower bound}
	\label{seccramerrao}
	
	\subsection{Background}
	\label{secCRBA}
	
	The classical Cram\'er-Rao bound (CRB) \cite{cramer1999mathematical} expresses a lower bound on the variance of any unbiased estimator of the support $\Omega$. In literature, the Cram\'er-Rao lower bound has been derived for various models of the support set, and we review some important contributions further below.

	Before we proceed, let us recall some old definitions and define some new ones that will be used in this section and Section \ref{sec:CRproof}. Let $\xi\in\T$. 
	\begin{itemize}
		\item 
		$\phi(\xi)\in\C^M$ denotes the vector whose $k$-th entry, for $0\leq k\leq M-1$, is $e^{-2\pi i k\xi}$.
		\item
		For each $\ell\geq 0$, we let $\phi^{(\ell)}\in\C^M$ be the vector whose entries are the $\ell$-th derivative of $\phi(\xi)$. Hence, $\phi^{(\ell)}(\xi)_k:=(-2\pi i k)^\ell e^{-2\pi ik\xi}$. 
		\item
		For convenience, we let $\psi(\xi):=\phi^{(1)}(\xi)$, which will play an important role below. 
		\item
		Given $\Omega$, we let $\Phi, \Psi\in \C^{M\times S}$ be the matrices whose columns are of the form $\phi(\theta)$ and $\psi(\theta)$, respectively, for each $\theta\in \Omega$.
	\end{itemize} 
	
	The CRB for spectral estimation has been extensively studied \cite{stoica1989music,stoica1990music,lee1992cramer,koochakzadeh2016cramer,ottersten1993exact}. We briefly review some important results. A fundamental result in \cite[Theorem 4.1]{stoica1989music} implies that, under appropriate assumptions, for any unbiased estimator $\hat\Omega:=\{\hat\omega_{j}\}_{j=1}^S$ of $\Omega:=\{\omega_{j}\}_{j=1}^S$, {not necessarily just MUSIC or ESPRIT}, we have
	\begin{equation}
	\begin{aligned}
		\label{eq:cov1}
		\E \Big( \big(\hat\Omega-\Omega\big)\big(\hat\Omega-\Omega\big)^*\Big) 
		\geq \frac{\nu^2}{2L} \( \text{Re}\big( \Psi^* (I-P_{\Phi}) \Psi \odot X \big)\)^{-1}.
	\end{aligned}
	\end{equation}
	To simplify the notation in this section, we always assume that the elements of $\hat\Omega$ have been re-indexed to minimize its matching distance to $\Omega$.

	While inequality \eqref{eq:cov1} already gives the optimal dependence on the noise variance $\nu^2$ and the number of snapshots $L$, the matrix that appears on the right hand side {implicitly depends on $\Omega$, so it is unclear how $( \text{Re}( \Psi^* (I-P_{\Phi}) \Psi \odot X \big))^{-1}$ behaves as a function of $\Omega$ and $M$.} One would hope for an explicit bound depending on the geometry of $\Omega$. 
	
	The paper \cite{lee1992cramer} considered a situation where $\Omega$ consists of $S$ points approximately spaced by $\epsilon$. The main result in the aforementioned paper provided an expansion of the right hand side of \eqref{eq:cov1} in the asymptotic limit $\epsilon\to0$. In particular, it was shown that 
	\begin{equation}
		\E \big( \md(\hat\Omega,\Omega)^2 \big) 
	    \geq \frac{C \nu^2 \epsilon^{-2S+2}}{L}  + O(\epsilon^{-2S+3}),
		\label{lee1992cramereq}
	\end{equation}
	where $C>0$ is independent of $\epsilon, \nu,L$ but is allowed to depend on the other parameters.

	\subsection{New CRB for clumps model}
	
	{The geometric results in Section \ref{secsuperresolution}, as well as prior work on super-resolution \cite{li2021stable}, strongly suggest that inequality \eqref{lee1992cramereq} is unnecessarily pessimistic and is achieved in the  worst-case scenario where all point sources are located in a single clump. 
	
	This section provides a new CRB for the separated clumps model, under the additional requirements that each clump has $\lambda$ elements that are equally spaced by a small parameter $\epsilon$. We will derive a CRB on the order of $\epsilon^{-2\lambda+2}$, which is much improved since it is possible that $\lambda\ll S$. 
	}
	
	We consider a situation where there are $R$ clumps that are far apart, and each clump contains $\lambda$ points separated by $\epsilon$. This is a more specific model than the one considered in Assumption \ref{separatedclumpsmodel}.
	
	\begin{definition}
		\label{def:Omegaep}
		We say $\Omega_\epsilon\subset\T$ is a $(\epsilon,R,\lambda,\{\theta_r\}_{r=1}^R)$ set if $\epsilon< \lambda^{-1}\min_{r\not=s} |\theta_r-\theta_s|_\T$ and $\Omega_\epsilon$ can be written as
		\[
		\Omega_\epsilon
		=\bigcup_{r=1}^R \{\theta_r,\theta_r+\epsilon,\dots,\theta_r+(\lambda-1)\epsilon\}.
		\]
		Note that this implies $\Omega_\epsilon$ consists of $R$ disjoint sets each supported in an interval of length $(\lambda-1)\epsilon$ and $\Omega_\epsilon$ has cardinality $S=R\lambda$. 
	\end{definition}
	
	We also require the following assumptions on the noise and amplitudes. 
	
	\begin{assumption}
		\label{assump:CR}
		Fix any positive integers $L,M,S$ such that $ M\ge S$ and $ L\ge S$.
		\begin{enumerate}[(a)]
			\item
			Fix $M,R,\lambda$ and distinct $\theta_1,\dots,\theta_R\subset\T$ and let $\Omega_\epsilon$ be a $(\epsilon,R,\lambda,\{\theta_r\}_{r=1}^R)$ set. 
			\item
			The amplitude covariance matrix $X$ is strictly positive-definite. 
			\item
			For each $t>0$ and $1\leq j\leq M$, $\e_j(t)$ is a Gaussian random vector with independent entries such that  $\e_j(t)\sim N(0,\nu^2)$. Also assume that $\e(s)$ and $\e(t)$ are independent for $s\not=t$.
			\item
			For each $1\leq \ell\leq L$, we are given
			\[
			\y(t_\ell)=\Phi(\Omega_\epsilon) \x(t_\ell) + \e(t_\ell).
			\]
		\end{enumerate}
	\end{assumption}

	\begin{theorem}
		\label{thm:CR}
		Under Assumption \ref{assump:CR}, there exists $\epsilon_0<(2\pi \lambda(M-1))^{-1}$ depending only on $M$, $\lambda$, and $\{\theta_r\}_{r=1}^R$ such that the following hold. For all $\epsilon\in (0,\epsilon_0)$,  any unbiased estimator $\hat\Omega_\epsilon$ of $\Omega_\epsilon$ satisfies the lower bound
		\begin{equation}
			\label{eq:CR}
			\E \big( \md(\hat\Omega_\epsilon,\Omega_\epsilon)^2\big)
		 	\geq \frac{C \, \big(\epsilon (M-1) \big)^{-2\lambda+2} \, \nu^2}{L(M-1)^3 \|X\|_2},
		 \end{equation}
		where the implicit constant $C>0$ does not depend on $L$, $M$, $X$, $\epsilon$, and $\nu$. 
	\end{theorem}
    Theorem \ref{thm:CR} is proved in Section \ref{sec:CRproof}. 
    
    The right hand side of \eqref{eq:CR} can be made more explicit by using the super-resolution factor. Since $\epsilon$ is the minimum separation of $\Omega_\epsilon$ and $M$ is the total number of Fourier measurements, we see that $\epsilon (M-1)=\SRF^{-1}$. Note that the assumption on $\epsilon_0$ in Theorem \ref{thm:CR} necessarily implies $\SRF>1$, and \eqref{eq:CR} becomes
		\begin{equation}
			\label{eq:CR2}
			\begin{aligned}
			\E \big( \md(\hat\Omega_\epsilon,\Omega_\epsilon)^2\big)
			\geq \frac{C \, \SRF^{2\lambda-2} \, \nu^2}{L(M-1)^3 \|X\|_2} . 
			\end{aligned}
		\end{equation}

\subsection{Comparison of our CRB and the classical CRB}

    In order to explain the main improvement of Theorem \ref{thm:CR} over classical the CRBs, consider the support set shown in Figure \ref{Fig1}. For the $\Omega$ depicted in Figure \ref{Fig1}, the classical CRB, such as \eqref{lee1992cramereq}, yields the prediction that for any unbiased estimator $\hat\Omega$,
    $$
    \E\big(\md(\hat\Omega,\Omega)^2\big) \geq \frac{C \epsilon^{-10} \nu^2}{L} + \text{lower order terms}.
    $$
    On the other hand, %for this same set, 
    our CRB in Theorem \ref{thm:CR} yields
	\begin{equation}
			\label{eq:CR3}
			\EE \big( \md(\widehat\Omega,\Omega)^2\big)
		 	\geq \frac{C  \, \nu^2  \varepsilon^{-4} }{(M-1)^{7} L \|X\|}. 
		 \end{equation}
		 In the super-resolution regime, $\varepsilon$ can be much smaller than $1/M$, so our CRB provides a significant improvement over the classical CRB.
	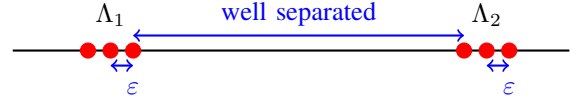
\begin{figure}[h]
	\centering
	
		\begin{tikzpicture}[xscale = 1,yscale = 1]
	\draw[thick] (1,0) -- (8.4,0);
	\filldraw[red] (2,0) circle (0.1cm);		
	\filldraw[red] (2.3,0) circle (0.1cm);		
	\filldraw[red] (2.6,0) circle (0.1cm);		
	\draw[blue,thick,<->] (2.3,-0.2) -- (2.6,-0.2);
	\node[blue,below] at (2.6,-0.3) {$\varepsilon$};
	\node[above] at (2.3,0.2) {$\Lambda_{1}$};
	\filldraw[red] (7,0) circle (0.1cm);		
	\filldraw[red] (7.3,0) circle (0.1cm);		
	\filldraw[red] (7.6,0) circle (0.1cm);		
	\draw[blue,thick,<->] (7.3,-0.2) -- (7.6,-0.2);
	\node[blue,below] at (7.6,-0.3) {$\varepsilon$};
	\node[above] at (7.3,0.2) {$\Lambda_2$};
	\draw[blue,thick,<->] (2.6,0.2) -- (7,0.2);
	\node[blue,above] at (4.8,0.2) {well separated};
	%%  
	%\node[below] at (-0.5,-1) {$\lambda =|\Lambda_1| = |\Lambda_2|= \ldots = |\Lambda_A| $};
	\end{tikzpicture}
	\caption{A comparison of the classical CRB and our new CRB for clumps. $\Omega$ consists of $2$ clumps where each clump contains $\lambda = 3$  equally spaced points with separation $\varepsilon$.}
	\label{Fig1}
\end{figure}

\subsection{Near optimality of ESPRIT}	

ESPRIT is one of the most popular subspace methods for spectral estimation. {To compare our upper bound for ESPRIT and our new CRB, we consider a situation where the assumptions to both Theorems \ref{thm:espritsuper1} and \ref{thm:CR} hold. 
\begin{itemize}
    \item 
    Suppose $\Omega$ consists of separated clumps with parameters $(R,M,S,\alpha,\beta)$ where $\beta$ satisfies the relationship \eqref{betaeq}. Moreover, we further impose that each clump has cardinality $\lambda$ and are equally spaced by $\alpha/(M-1)$. Then $\Omega$ is also a $(\epsilon,R,\lambda, \{\theta_r\}_{r=1}^R)$ set where $\epsilon=\alpha/(M-1)$ and each $\theta_r$ is the left most point in the $r$-th clump. An example is depicted in Figure \ref{FigDemoClumps1}.
    \item
    Assume that the amplitude covariance matrix $X$ is strictly positive-definite. 
    \item 
    Assume that for each $t>0$ and $1\leq j\leq M$, $\e_j(t)$ is a Gaussian random vector with independent entries such that  $\e_j(t)\sim N(0,\nu^2)$. Also assume that $\e(s)$ and $\e(t)$ are independent for $s\not=t$.
\end{itemize}
If $L/\nu^2$ is sufficiently large so that the assumptions in Theorem \ref{thm:espritsuper1} part (b) hold, then the output of ESPRIT satisfies 
$$
\E \big( \md(\hat\Omega, \Omega)^2\big)
\lesssim_{M,S,\lambda} \frac{\SRF^{2\lambda-2}}{\lambda_S(X)} \frac{\nu^2}{L}. 
$$
On the other hand, for any unbiased estimator $\hat\Omega$, Theorem \ref{thm:CR} and \eqref{eq:CR2} demonstrate that
$$
\E \big( \md(\hat\Omega, \Omega)^2\big)
\gtrsim_{M,S,\lambda} \frac{\SRF^{2\lambda-2}}{\|X\|_2} \frac{\nu^2}{L}.
$$
Notice that these statements are not contradictory since $\|X\|_2=\lambda_1(X)\geq \lambda_S(X).$

Thus, we observe that ESPRIT has the same dependence on the noise level $\nu$, the number of snapshots $L$, and the super-resolution factor $\SRF$ compared to the optimal unbiased estimator(s). 
}

\section{Numerical experiments}\label{section:num_exp}
 In this section, we perform systematic numerical simulations to validate our theory, under the separated clumps model. We consider the support $\Omega$ consisting of $2$ clumps of point sources, i.e., $R=2$. Each clump contains $\lambda$ equally spaced points separated by $\Delta$ for $\lambda =2$ and $\lambda = 3$ respectively.  Each amplitude is i.i.d. complex normal $CN(0,1)$ such that the real and imaginary parts  are independent normal, i.e.  $N(0,1/2)$. With such amplitudes, $\lambda_S(X) \approx 1$. Each noise is an i.i.d. $CN(0,\nu^2)$ random variable whose real and imaginary parts  are independent normal, i.e.  $N(0,\nu^2/2)$. We take $L$ snapshots of $M =100$ noisy Fourier measurements according to \eqref{eqlinearsystem}. Our upper bounds are proved for the noise-space correlation (NSC)  function perturbation in MUSIC and the support error in ESPRIT. When $\lambda=2$ and $\lambda =3$, we conduct $100$ and $2500$ trials of experiments respectively, and display the average NSC function perturbation in MUSIC and support error in ESPRIT, as well as the standard deviation in Figure \ref{FigErrorversusNoise} - \ref{FigErrorversusSRF}.
 The standard deviation is represented by the error bar.

Our numerical experiments in Figure \ref{FigErrorversusNoise} - \ref{FigErrorversusSRF} are to validate the NSC perturbation in MUSIC and the  support error in ESPRIT versus the noise standard deviation $\nu$, the number of snapshots $L$, $\sigma_S(\Phi)$ and $\SRF$. For MUSIC, the perturbation of the NSC is proved in Theorem \ref{thmmusicstability} and Theorem \ref{thm:musicsuper1} where the dependence on $\nu$, $L$, $\sigma_S(\Phi)$, $\SRF$, $M$ can be summarized as:
\begin{equation}
\begin{aligned}
	\E \big( \|\hat\calR -\calR\|_\infty \big)
%	\le
%	\sqrt{\E \big( \|\hat\calR -\calR\|_\infty^2 \big)}
	\lesssim \frac{\sqrt{ M} \nu}{ \sqrt{\lambda_S(X)} \sigma_S(\Phi)\sqrt{L}}
	\lesssim \frac{\sqrt{ M} \, \SRF^{\lambda-1}\nu}{ \sqrt{\lambda_S(X)}\sqrt{L}} .
	\end{aligned}
		\label{expmusic1}
		\end{equation}
		%The dependence on $\nu$, $L$, $\sigma_S(\Phi)$, $\SRF$ is given in \eqref{thmmusicstabilityeq}, \eqref{thm:musicsuper1eq1} and \eqref{eqmusicscaling}. 
For ESPRIT, the support error is proved in Theorem \ref{thmespritstability} and Theorem \ref{thm:espritsuper1} where the dependence on $\nu$, $L$, $\sigma_S(\Phi)$, $\SRF$, $M$ in the large SNR regime can be summarized as:  \begin{equation}
	\E \big( {\rm md}(\hat\Omega,\Omega)\big) 
			\lesssim \frac{\sqrt{ M} \nu}{ \sqrt{\lambda_S(X)} \sigma_S(\Phi)\sqrt{L}} 
	\lesssim \frac{\sqrt{ M} \, \SRF^{\lambda-1}\nu}{ \sqrt{\lambda_S(X)}\sqrt{L}} .
		\label{expesprit1}
		\end{equation} 
Finally, we show phase transition figures for ESPRIT that validate our perturbation analysis for ESPRIT.

\subsection{Error versus Noise}

We first test the NSC perturbation in MUSIC and the support error in ESPRIT versus Noise.   We set $\SRF =5$, $L=1000$ and $25000$ for $\lambda=2$ and $3$ respectively.
 Figure \ref{FigErrorversusNoise} (a) shows the NSC perturbation in MUSIC as a function of noise $\nu$ in $\log_{10}$ scale. Figure \ref{FigErrorversusNoise} (b) shows the support error in ESPRIT as a function of noise $\nu$ in $\log_{10}$ scale. 
 The slope for all curves is close to $1$, which demonstrates that the NSC perturbation in MUSIC and  the support error in ESPRIT is linear in noise as predicted in \eqref{expmusic1} and \eqref{expesprit1}.

\begin{figure}[h]
\centering
\subfigure[Sensitivity of MUSIC to Noise]{
\includegraphics[width=0.23\textwidth]{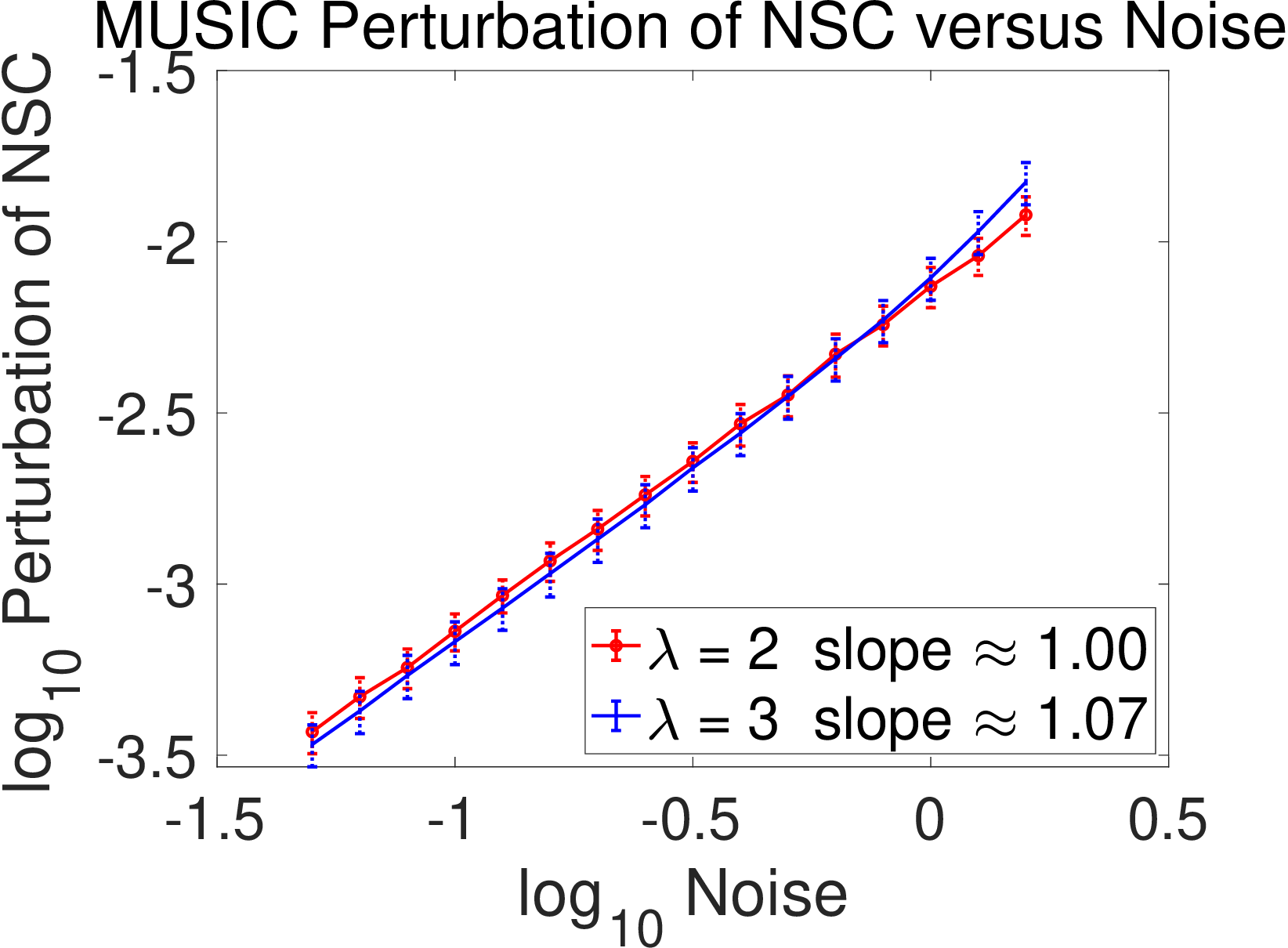}
}
%\caption{fig1}}
\subfigure[Sensitivity of ESPRIT to Noise]{
\includegraphics[width=0.22\textwidth]{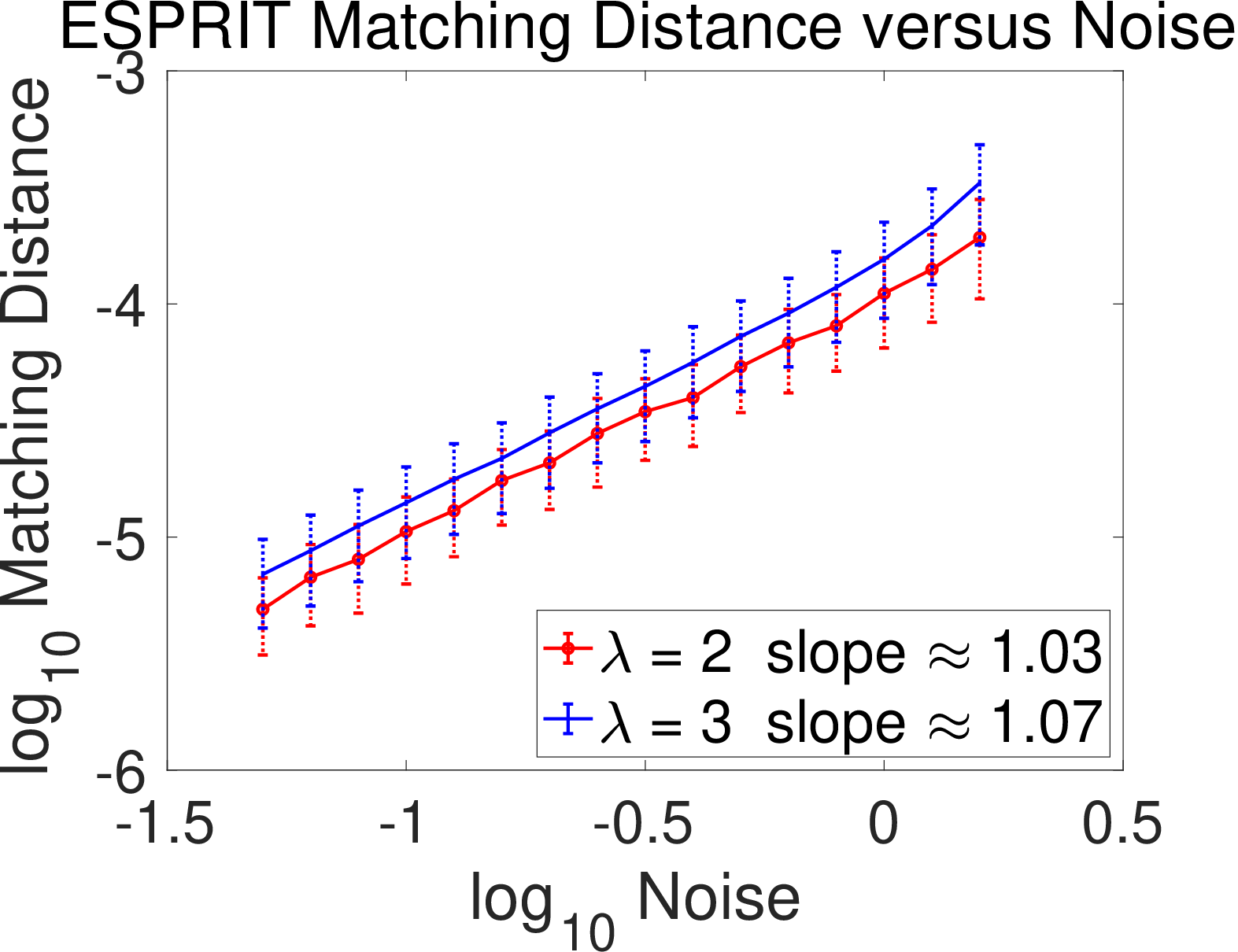}
}
\caption{(a) Noise-space correlation (NSC) function perturbation in MUSIC versus noise $\nu$ in $\log_{10}$ scale for $\lambda = 2$ and $\lambda = 3$; (b) Support error in ESPRIT versus noise $\nu$ in $\log_{10}$ scale for $\lambda = 2$ and $\lambda = 3$.  }
\label{FigErrorversusNoise}
\end{figure}

\subsection{Error versus the number of snapshots $L$}
We next test the NSC perturbation in MUSIC and  the support error in ESPRIT versus the number of snapshots $L$, when $\nu=0.1$ and $\SRF =5$. 
Figure \ref{FigErrorversusL} (a) shows the NSC perturbation in MUSIC as a function of $L$ in $\log_{10}$ scale. Figure \ref{FigErrorversusL} (b) shows the support error in ESPRIT as a function of $L$ in $\log_{10}$ scale. 
 The slope for all curves is close to $-0.5$, which is consistent with \eqref{expmusic1} and \eqref{expesprit1}.
 %demonstrates that the noise-space correlation function perturbation in MUSIC and the support error in ESPRIT scales as $1/\sqrt L$.

\begin{figure}
\centering
\subfigure[Sensitivity of MUSIC to $L$]{
\includegraphics[width=0.23\textwidth]{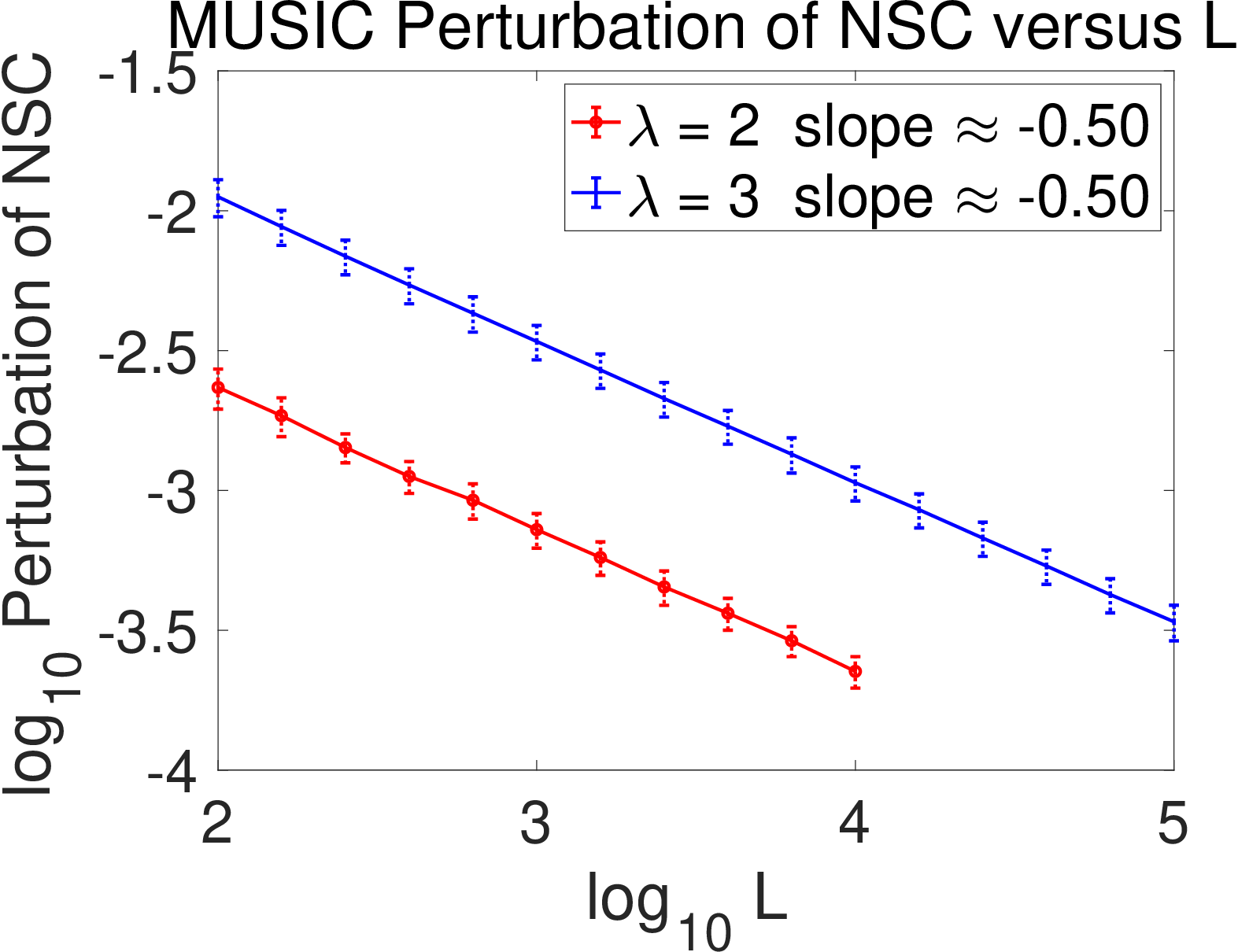}}
\subfigure[Sensitivity of ESPRIT to $L$]{\includegraphics[width=0.22\textwidth]{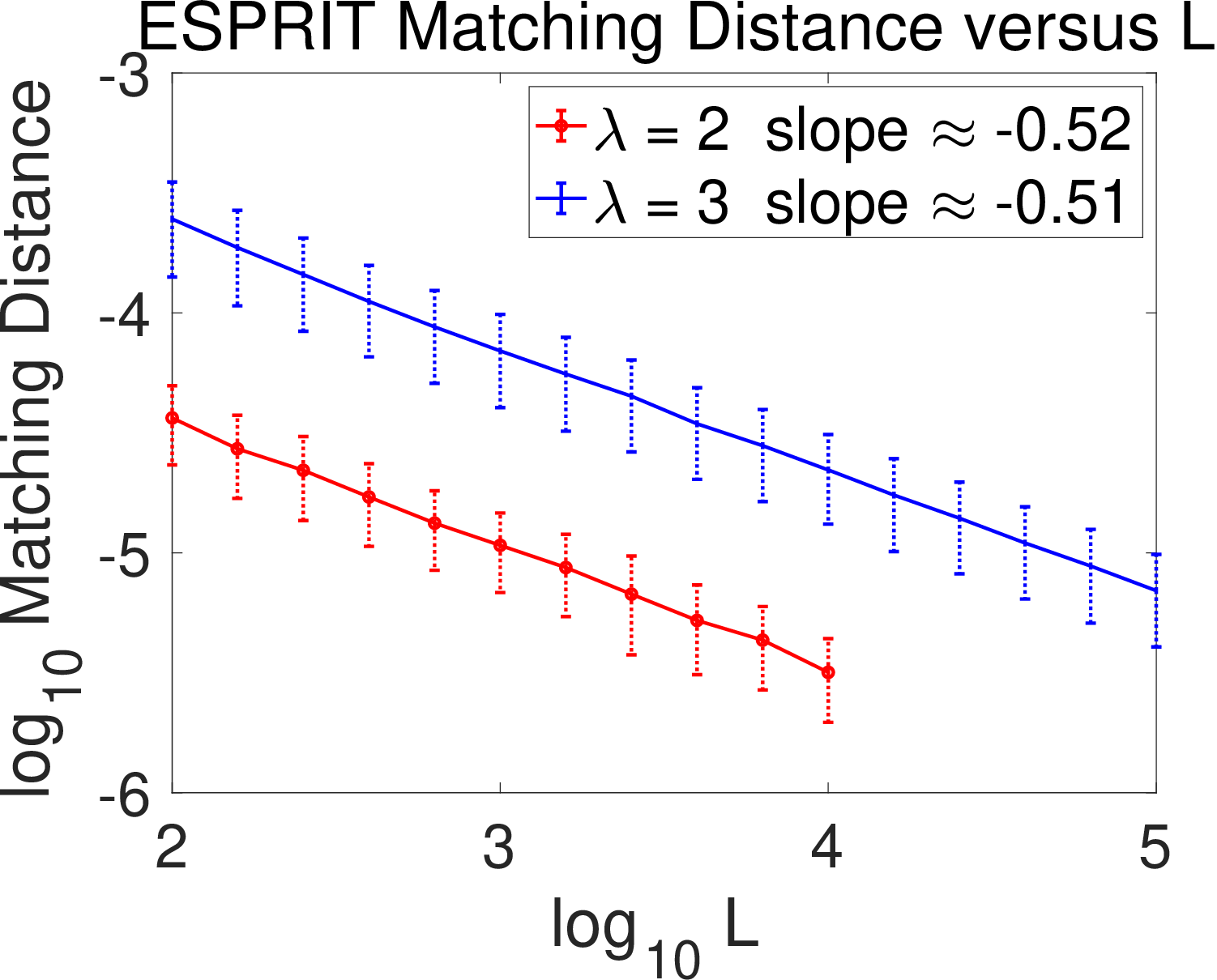}}
\caption{(a) NSC perturbation in MUSIC versus $L$ in $\log_{10}$ scale for $\lambda = 2$ and $\lambda = 3$; (b) Support error in ESPRIT versus $L$ in $\log_{10}$ scale for $\lambda = 2$ and $\lambda = 3$. }
\label{FigErrorversusL}
\end{figure}

\subsection{ Error versus $\sigma_S(\Phi)$ and super-resolution factor (SRF)}
Finally we test the NSC perturbation in MUSIC and  the support error in ESPRIT versus $\sigma_S(\Phi)$ and SRF respectively, when $L=1000$ and $\nu =0.1$. While $R$ and $\lambda$ are fixed, we vary $\Delta$ such that $\sigma_S(\Phi)$ and ${\rm SRF}$ vary.

Figure \ref{FigErrorversusSigmaSPhi} (a) shows the NSC perturbation in MUSIC as a function of $\sigma_S(\Phi)$ in $\log_{10}$ scale. Figure \ref{FigErrorversusSigmaSPhi} (b) shows the support error in ESPRIT as a function of $\sigma_S(\Phi)$ in $\log_{10}$ scale. 
The slope for all curves is around $-1$, which is consistent with \eqref{expmusic1} and \eqref{expesprit1}.

Figure \ref{FigErrorversusSRF} (a) shows the NSC perturbation in MUSIC as a function of ${\rm SRF}$ in $\log_{10}$ scale. Figure \ref{FigErrorversusSRF} (b) shows the support error in ESPRIT as a function of ${\rm SRF}$ in $\log_{10}$ scale. 
The slope for all curves is around $\lambda-1$, which is consistent with \eqref{expmusic1} and \eqref{expesprit1}.

\begin{figure}[h]
\centering
\subfigure[Sensitivity of MUSIC to $\sigma_S(\Phi)$]{
\includegraphics[width=0.225\textwidth]{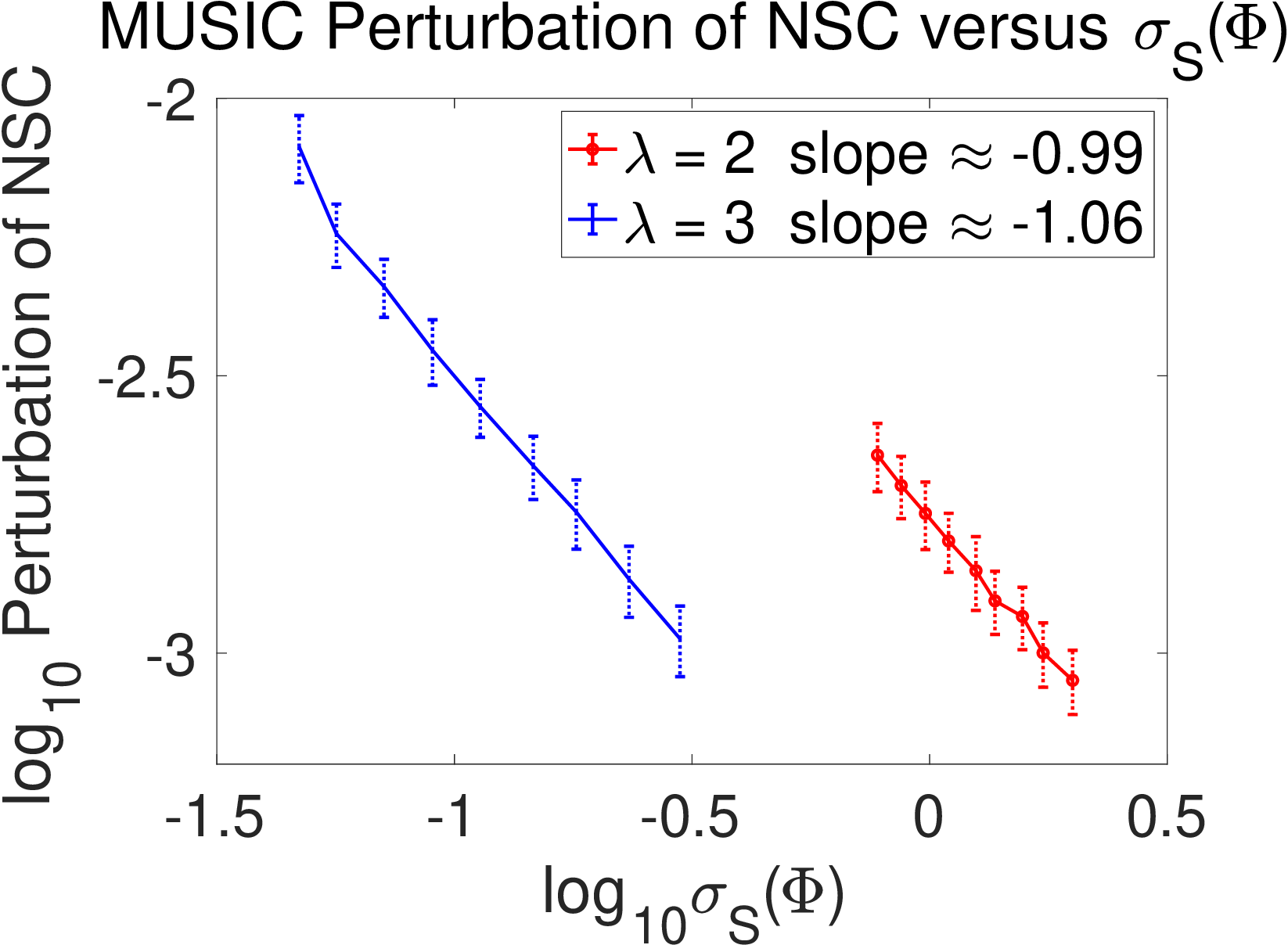}
}
\subfigure[Sensitivity of ESPRIT to $\sigma_S(\Phi)$]{
\includegraphics[width=0.225\textwidth]{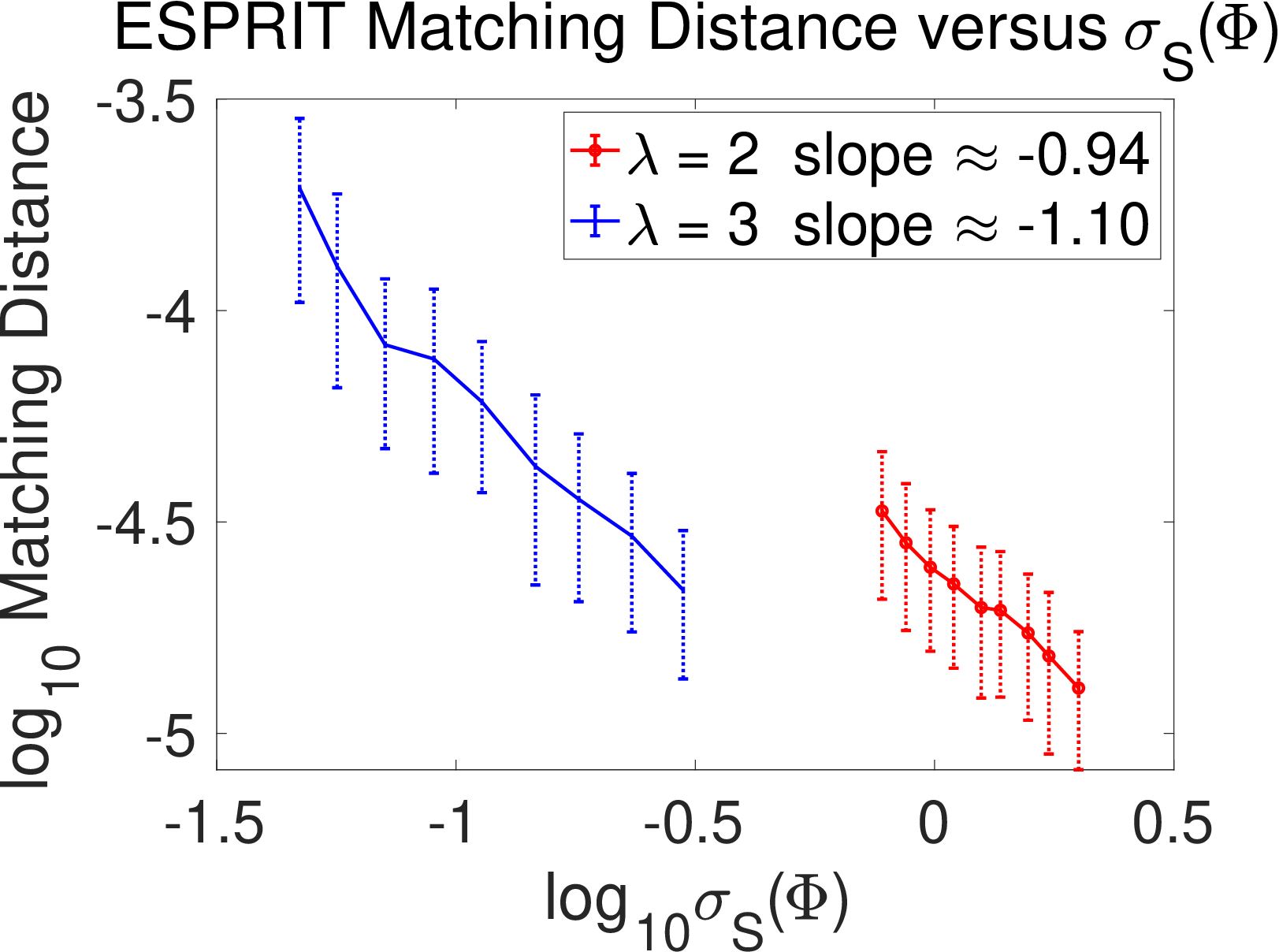}}
%\caption{fig1}}
\caption{(a) NSC perturbation in MUSIC versus $\sigma_S(\Phi)$ in $\log_{10}$ scale for $\lambda = 2$ and $\lambda = 3$; (b) Support error in ESPRIT versus $\sigma_S(\Phi)$ in $\log_{10}$ scale for $\lambda = 2$ and $\lambda = 3$. 
}
\label{FigErrorversusSigmaSPhi}
\end{figure}

\begin{figure}[h]
\centering
\subfigure[Sensitivity of MUSIC to ${\rm SRF}$]{
\includegraphics[width=0.22\textwidth]{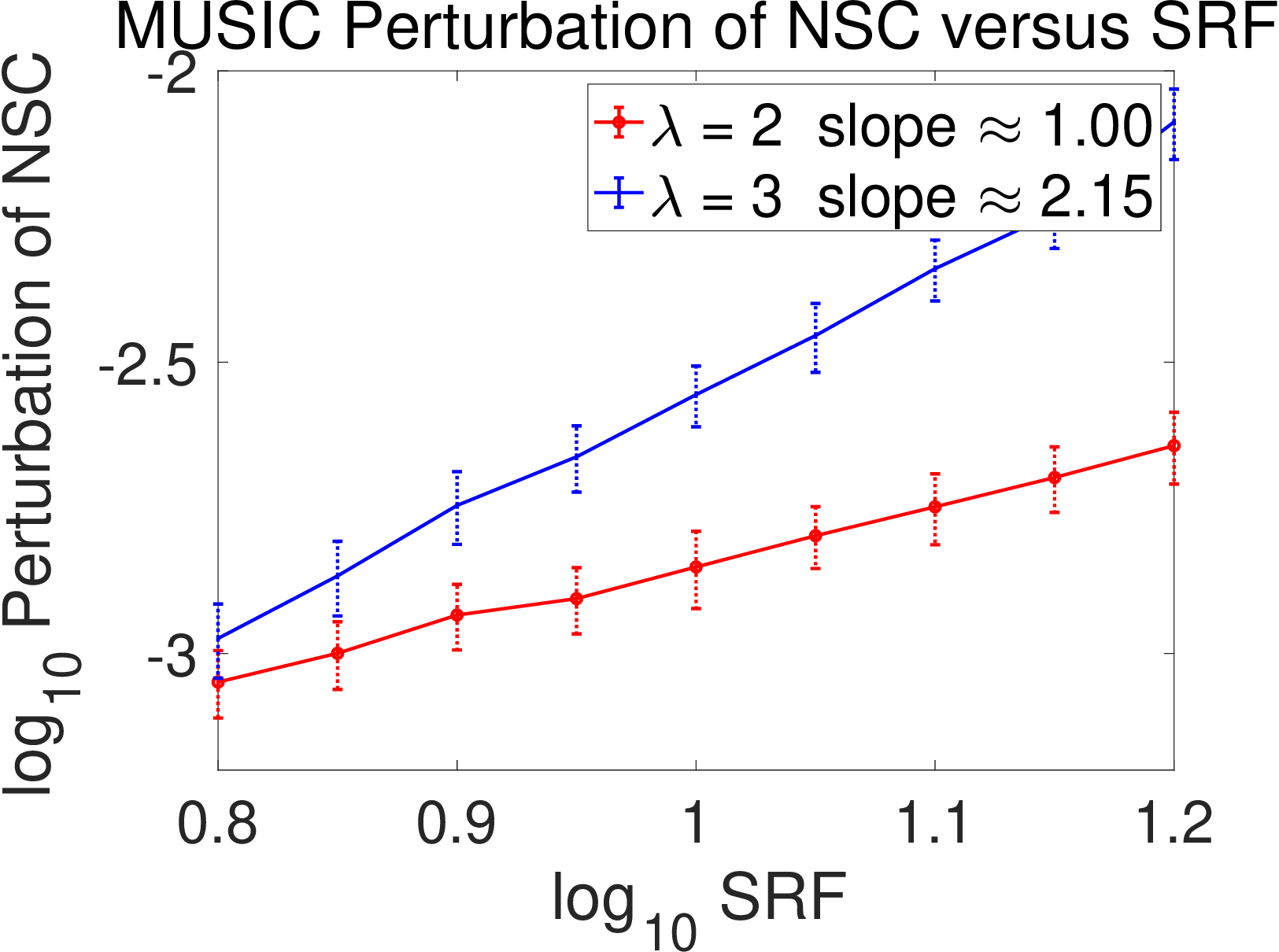}
}
\subfigure[Sensitivity of ESPRIT to ${\rm SRF}$]{
\includegraphics[width=0.22\textwidth]{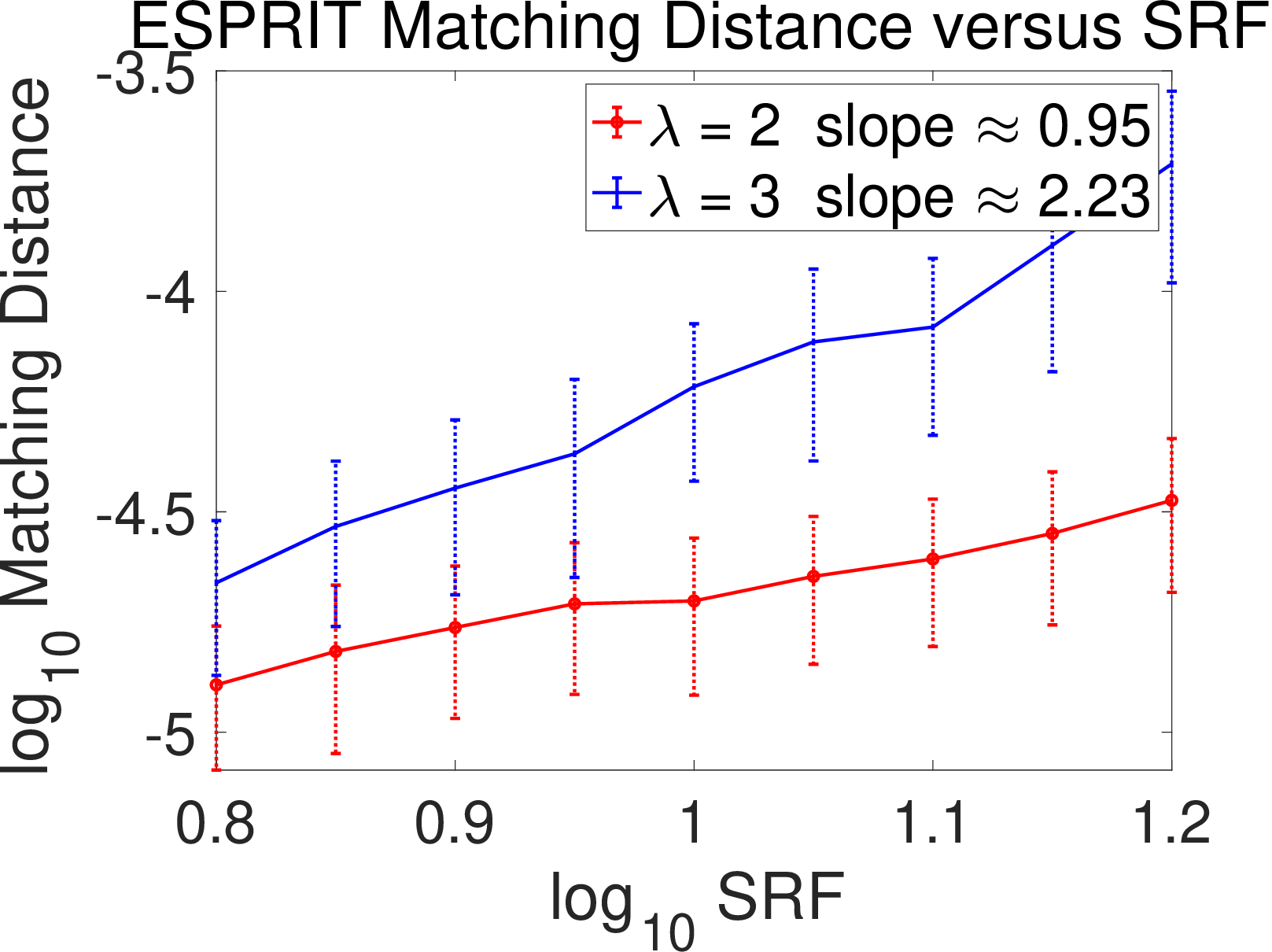}}
%\caption{fig1}}
\caption{(a) NSC perturbation in MUSIC versus ${\rm SRF}$ in $\log_{10}$ scale for $\lambda = 2$ and $\lambda = 3$; (b) Support error in ESPRIT versus ${\rm SRF}$ in $\log_{10}$ scale for $\lambda = 2$ and $\lambda = 3$. 
}
\label{FigErrorversusSRF}
\end{figure}

{

\subsection{Phase transition figures for ESPRIT}

Finally we show the phase transition figures of ESPRIT in Figure \ref{FigPhasetransition}.  Figure \ref{FigPhasetransition} shows the average $\log_2[{\rm md}(\hat\Omega,\Omega)/\Delta]$ with respect to: (a,b) $\log_{10} L $ and $\log_{10} \nu $ when $\lambda = 2$ and $\lambda =3$; (c,d) $\log_{10} L $ and $\log_{10} \SRF $ when $\lambda = 2$ and $\lambda =3$; (e,f) $\log_{10} \SRF $ and $\log_{10} \nu $ when $\lambda = 2$ and $\lambda =3$.
We observe a clear phase transition phenomenon, and a blue phase transition curve is extracted in each figure. According to \eqref{expesprit1}, we expect the phase transition curve follows
\begin{equation}(\lambda-1)\log_{10} \SRF + \log \nu \sim 1/2 \log L.
\label{expphasetransitionscaling}
\end{equation}
After numerically fitting the phase transition curves, our empirical slope is consistent with our theoretical one in \eqref{expphasetransitionscaling}.
}

\begin{figure}[h]
\centering
\subfigure[Noise-$L$ Phase transition when $\lambda=2$, slope $\approx 0.50$]{
\includegraphics[width=0.22\textwidth]{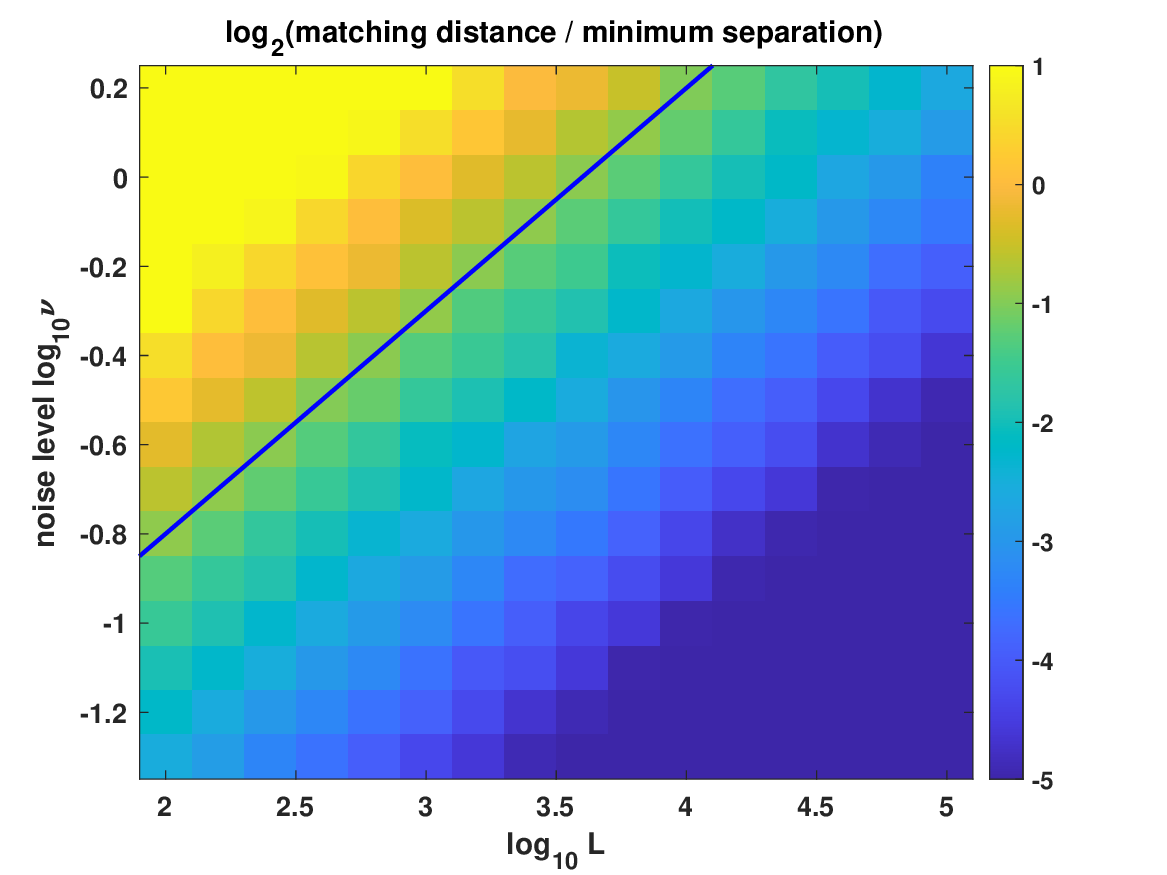}
}
\subfigure[Noise-$L$ Phase transition when $\lambda=3$, slope $\approx 0.49$]{
\includegraphics[width=0.22\textwidth]{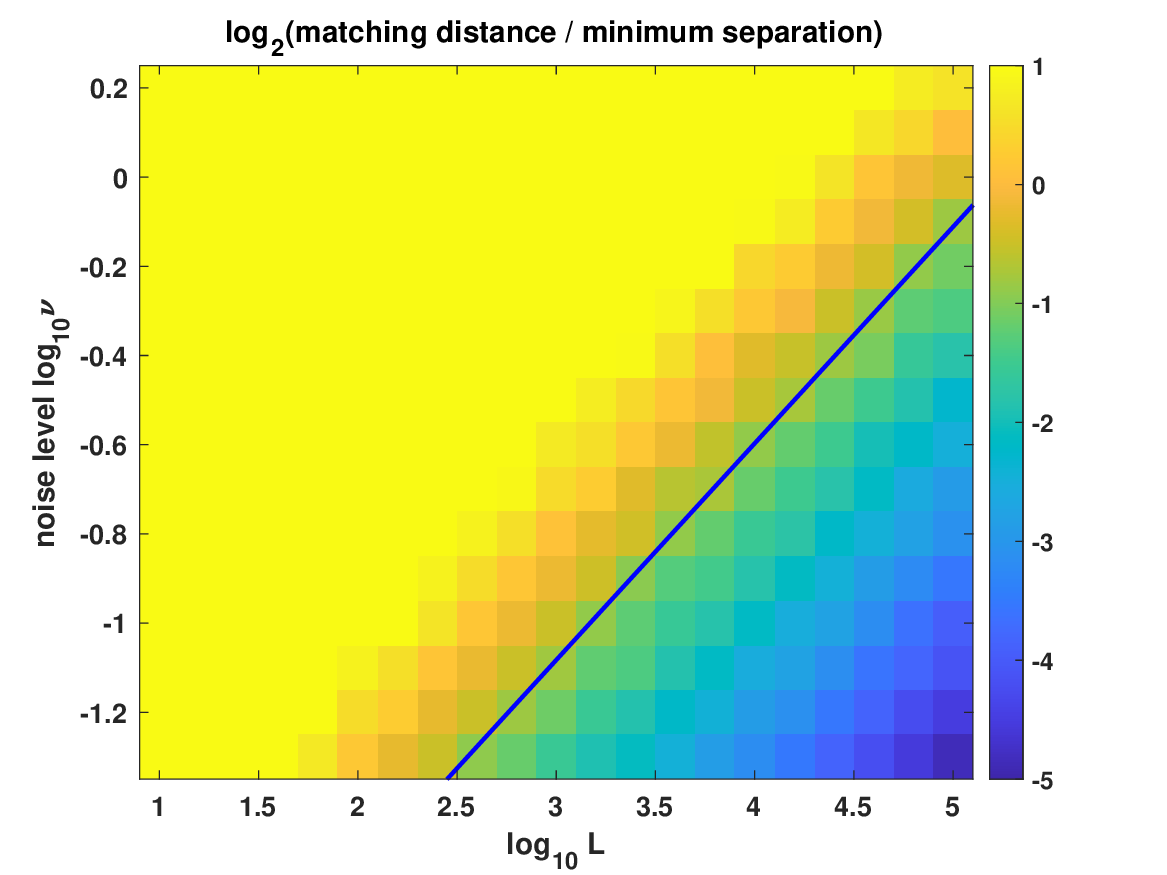}}
\subfigure[SRF-$L$ Phase transition when $\lambda=2$, slope $\approx 0.50$]{
\includegraphics[width=0.22\textwidth]{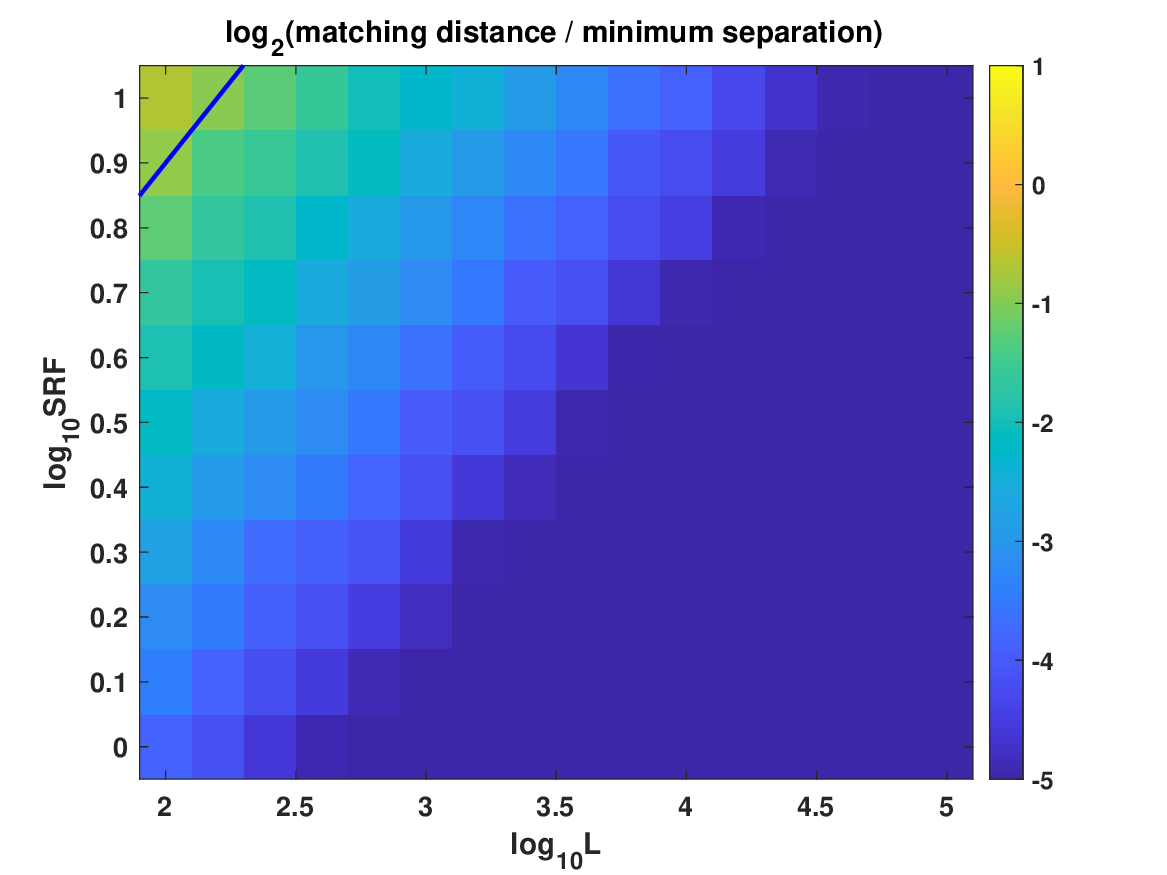}
}
\subfigure[SRF-$L$ Phase transition when $\lambda=3$, slope $\approx 0.25$]{
\includegraphics[width=0.22\textwidth]{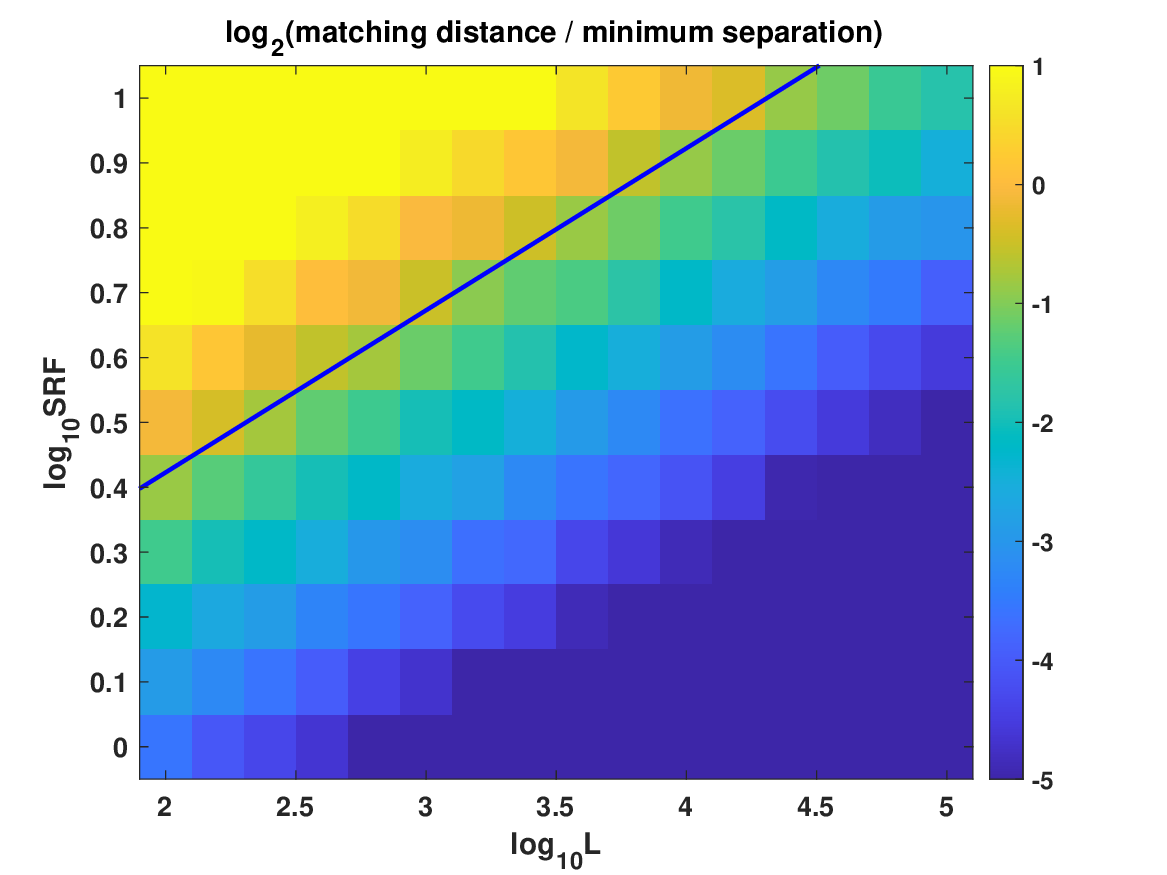}}
\subfigure[Noise-SRF Phase transition when $\lambda=2$, slope $\approx -0.98$]{
\includegraphics[width=0.22\textwidth]{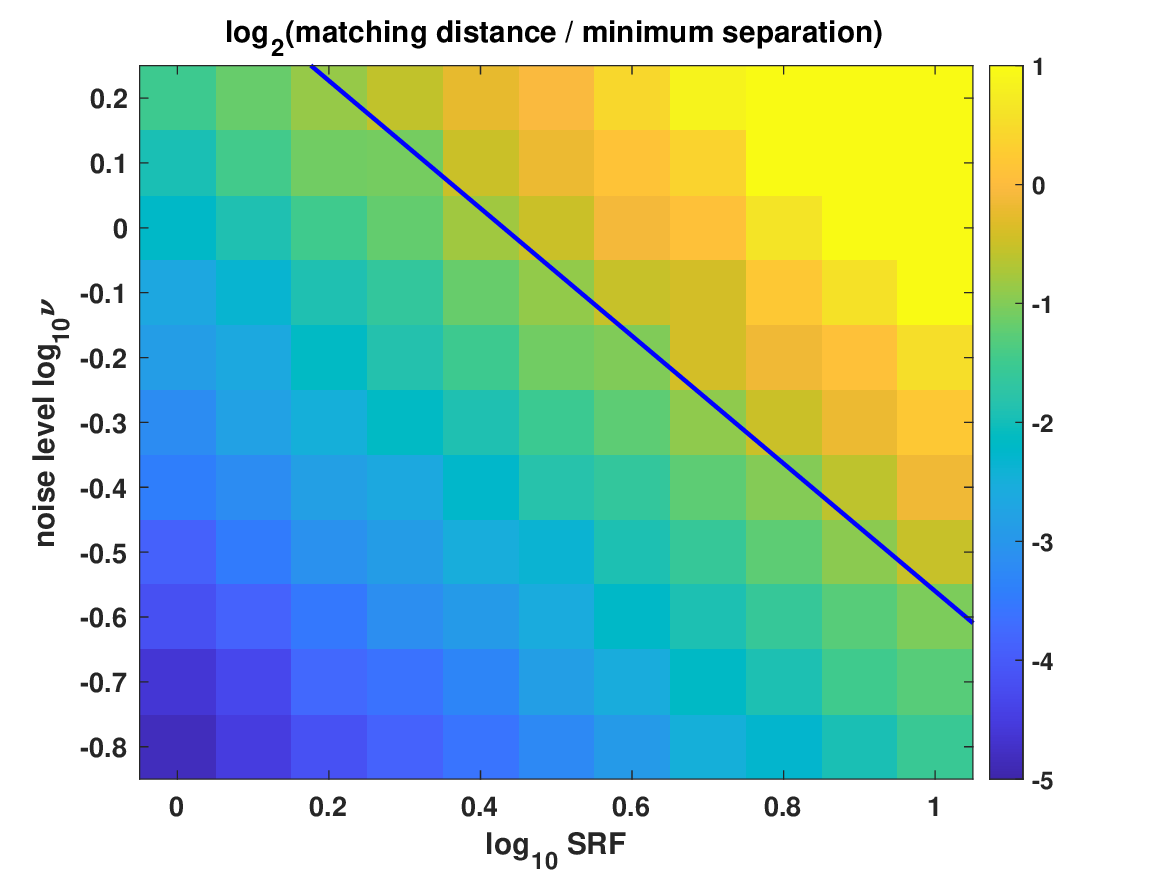}
}
\subfigure[Noise-SRF Phase transition when $\lambda=3$, slope $\approx -1.93$]{
\includegraphics[width=0.22\textwidth]{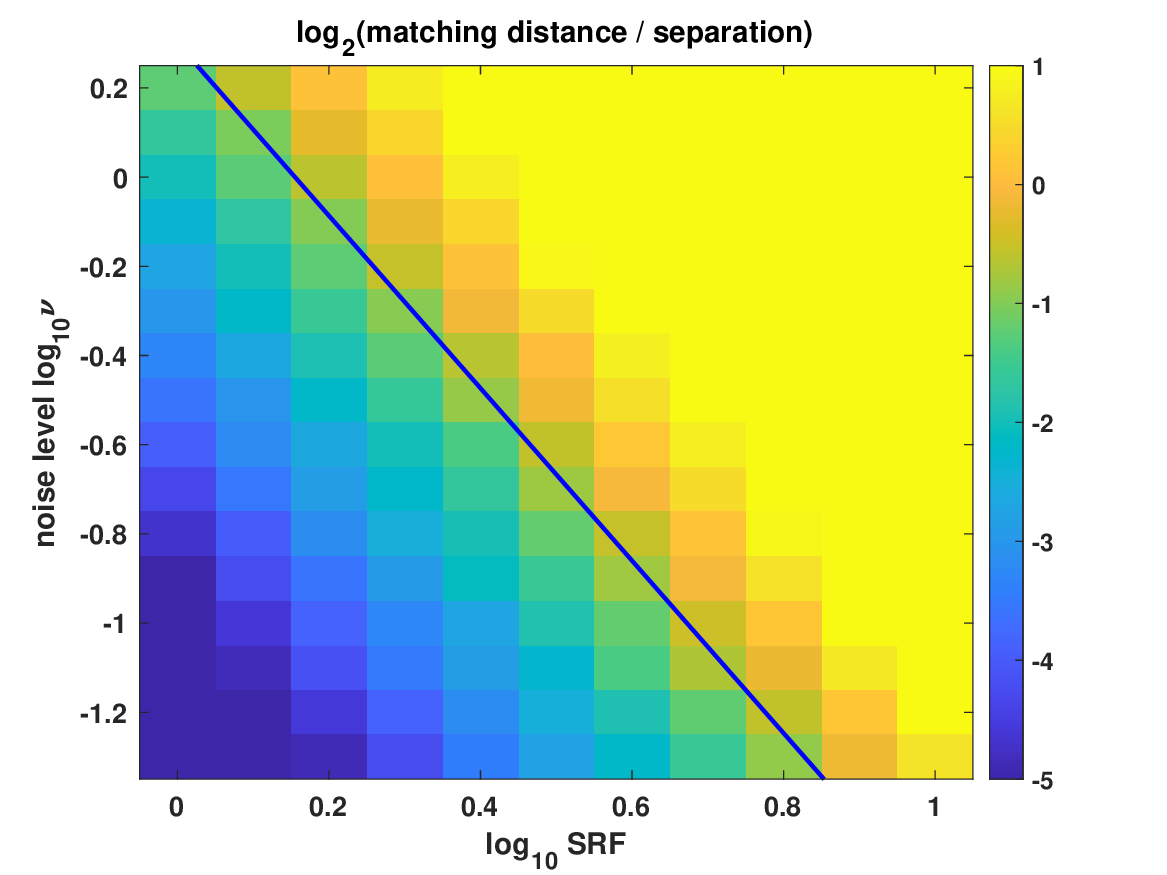}}
\caption{Color plot for the average $\log_2[{\rm md}(\hat\Omega,\Omega)/\Delta]$ for ESPRIT with respect to: (a)-(b)  $\nu$ vs $L$, (c)-(d) SRF vs $L$, (e)-(f) $\nu$ vs SRF when $\lambda =2$ (left column) and $\lambda=3$ (right column). The slope in (a,b) is about 0.5; the slope in (c,d) is about $0.5/(\lambda-1)$; the slope in (e,f) is about $-(\lambda-1)$, which is consistent with \eqref{expphasetransitionscaling}.}

%(a)  $\log_{10} L$ in horizontal axis and  $\log_{10} \nu $ in vertical axis when $\lambda = 2$; (b) $\log_{10} L $ in horizontal axis and $\log_{10} \nu $ in vertical axis when $\lambda =3$; (c) $\log_{10} L $ in horizontal axis and $\log_{10} \SRF $ in vertical axis when $\lambda = 2$; (d) $\log_{10} L $ in horizontal axis and $\log_{10} \SRF $ in vertical axis when $\lambda = 3$; (e) $\log_{10} \SRF $ in horizontal axis and $\log_{10} \nu $ in vertical axis when $\lambda = 2$; (f) $\log_{10} \SRF $ in horizontal axis and $\log_{10} \nu $ in vertical axis when $\lambda = 3$. In (a) and (b), the slope of the phase transition curve is about $0.5$, which confirms $ \log \nu \sim 1/2 \log L$  \eqref{expphasetransitionscaling}. In (c) and (d), the slope of the phase transition curve is about $0.5/(\lambda-1)$, which confirms  \eqref{expphasetransitionscaling}.

\label{FigPhasetransition}
\end{figure}

\section{Multiple versus single snapshot}
\label{sec:singlevsmultiple}

This paper focuses on the multi-snapshot spectral estimation problem where Fourier measurements of $\x(t)$ and $\Omega$ are collected at $L$ different time instances. Since information is collected at different times, it is natural to impose statistical assumptions, which is what we and many other papers have done. 

MUSIC and ESPRIT are still applicable for the single-snapshot spectral estimation problem (the $L=1$ case), with proper modification. Instead of forming the empirical covariance matrix $\hat Y$ in the first step of Algorithm \ref{algmusicesprit}, one forms the Hankel matrix of a single-snapshot \cite{liao2016music,li2020super}. From a theoretical perspective, the single-snapshot problem is usually treated from a deterministic point of view. 

Single-snapshot MUSIC and ESPRIT enjoy theoretically provable robustness properties in the super-resolution regime, as shown in \cite{li2021stable,li2020super}. The main results there show that for MUSIC and ESPRIT,
\begin{equation}
\text{single-snapshot error}
\lesssim \frac{\nu}{\sigma_S^2(\Phi)}
\lesssim \SRF^{2(\lambda-1)}\nu, 
\end{equation}
where the implicit constants do not depend on the noise standard deviation $\nu$, the support set $\Omega$, and $\SRF$. This upper bound matches, up to implicit constants, the minimax lower bounds \cite{li2021stable,batenkov2020conditioning}.

A na\"ive guess is that the error incurred by multi-snapshot MUSIC and ESPRIT would be $1/\sqrt L$ times that of the single-snapshot setting. Perhaps surprisingly, our results show that this na\"ive guess is incorrect. Indeed, the main results of this paper for MUSIC (Theorems \ref{thmmusicstability} and \ref{thm:musicsuper1}) and ESPRIT (Theorems \ref{thmespritstability} and \ref{thm:espritsuper1}) show that on average,
\begin{equation}
\text{multi-snapshot error}
\lesssim \frac{\nu}{\sigma_S(\Phi)\sqrt L }
\lesssim \frac{\SRF^{\lambda-1}\nu}{\sqrt L}, 
\end{equation}
where again, the implicit constants do not depend on the noise standard deviation $\nu$, the support set $\Omega$, $\SRF$, and $L$. 

Comparing the upper bounds for single-snapshot and multi-snapshot errors,  we see that not only does collecting multiple snapshots provide us with the extra $1/\sqrt L$ cancellation, but it also enjoys an additional square root in the $1/\sigma_S(\Phi)$ and $\SRF$ terms. In other words, the condition numbers of multi-snapshot MUSIC and ESPRIT are significantly better than their single-snapshot counterparts. 

The improvement can be explained by the strictly positive-definite assumption on $X$ made in Assumption \ref{assump:main}. To see why, consider the deterministic setting where $\x(t)$ is fixed throughout time, so $\x(t_1)=\cdots=\x(t_L)$. Then the collected snapshots are of the form,
\[
\y(t_\ell)=\Phi(\Omega) \x(t_1) + \e(t_\ell). 
\]
Hence, collecting multiple snapshots does not provide us with any new information about the range of $\Phi$ since $\x(t)$ is constant, and the best that can be done is to average all $L$ such snapshots,
\[
\frac{1}{L} \sum_{\ell=1}^L \y(t_\ell)
=\Phi(\Omega) \x(t_1) + \frac{1}{L} \sum_{\ell=1}^L \e(t_\ell),
\]
in order to reduce the variance of the noise by a factor of $L$ and feed this averaged signal into single-snapshot MUSIC and ESPRIT. In this highly degenerate situation, the na\"ive guess for the stability of multi-snapshot MUSIC and ESPRIT is correct.

We also explain how the strictly positive-definite assumption on $X$ changes the inherent geometry of the problem. For the sake of this discussion, assume that $\e(t)=0$ for all $t>0$. Since $\Omega$ does not change in time, $\y(t_\ell)\in R(\Phi)$ for each $1\leq \ell\leq L$. If $X$ is strictly positive definite, then $\{\x(t_\ell)\}_{\ell=1}^L$ spans a $S$-dimensional subspace, so the span of $\{\y(t_\ell)\}_{\ell=1}^L$ is precisely $R(\Phi)$. On the other hand, if $X$ does not have full rank, then $\{\y(t_\ell)\}_{\ell=1}^L$ only contains partial information about the range of $\Phi$. In which case, multi-snapshot MUSIC and ESPRIT would fail and we would need to use their single-snapshot versions. 

\vspace{-.5em}

\section{Conclusion and discussion}

This paper analyzes the performance of MUSIC and ESPRIT, and identifies some key quantities that control their stability: number of snapshots, noise variance, amplitude covariance matrix, and the smallest singular value of Fourier matrices (which can then be further estimated in terms of SRF). Our results accurately quantify intuitive phenomenon. For instance, if many point sources are closely spaced then we expect recovery to be sensitive to noise -- this is precisely captured by our results since in this case, $\sigma_{S}(\Phi)$ would be small, and $\SRF^\lambda$ would be large. Likewise if the amplitudes of $\x(t)$ vary slowly over time, then we expect inversion to be challenging -- this is again captured by our results since in that case, $\lambda_S(X)$ would be small and lead to instability. If $\lambda_S(X)=0$, then multi-snapshot ESPRIT and MUSIC would fail, but we can instead use their single snapshot versions.

The upper bound for ESPRIT matches a Cram\'er-Rao lower bound in terms of SRF and max clump cardinality. This proves that no unbiased algorithm can be fundamentally better than ESPRIT for this imaging problem in the large SRF regime -- any improvements can only be made in terms of other model parameters such as $S$, $M$, and $\lambda$. These are valuable results for practitioners, as they provide an accurate and solid theoretical footing for discriminating between favorable and unfavorable imaging situations. While this paper only focuses on a one dimensional scenario with plain MUSIC and ESPRIT, we believe the techniques here are also relevant for the analysis of their multi-dimensional counterparts and variations.

\section{Acknowledgement}

All authors would like to thank Zai Yang for finding an error in the first version of this manuscript, and bringing \cite{cai2018rate} to our attention.

\section{Proofs}
\label{sec:proofs}

\subsection{Proof of Theorem \ref{thm:Uperturb}} 

\label{secproofthm:Uperturb}
    
	{ 
    \begin{proof}
	Recall that $Y_L = \Phi X_L + E_L$. To put this in the framework of \cite{cai2018rate}, we rescale the problem so that 
	$$
	\frac{\sqrt L}{\nu} Y_L = \frac{\sqrt L}{\nu} \Phi X_L + \frac{\sqrt L}{\nu} E_L, 
	$$
	so that now $E_L \sqrt L/\nu$ is a $M\times L$ matrix with i.i.d. $\mathcal{G}_{1,\tau}$ entries. This scaling does not change the singular spaces of $Y_L$ and $\Phi X_L$. Applying \cite[Theorem 3]{cai2018rate}, which extends to complex matrices without issue, there is a $C_\tau>0$ only depending on $\tau$ such that  
	\begin{align}
		\label{eq:Edist}
		\E \big( \dist(\hat U,U)^2 \big)
		\leq \frac{C_\tau M(\nu^2\sigma_S^2(\Phi X_L)+\nu^4)}{L\sigma_S^4(\Phi X_L)}. 
	\end{align}
	Since $X_L$ is full rank and assumption \eqref{eq:nu} holds, we have 
	$$
	\sigma_S(\Phi X_L)
	\geq \sigma_S(\Phi)\sigma_S(X_L)
	=\sigma_S(\Phi) \sqrt{\lambda_S(X)}
	\geq \nu. 
	$$
	Using this inequality in \eqref{eq:Edist} completes the proof. 
	\end{proof}
}

 \subsection{Proof of Theorem \ref{thmmusicstability}}
 \label{proofthmmusicstability}
    
    \begin{proof}
	We start with a deterministic inequality that links the perturbation of $\calR$ to the distance between $U$ and $\hat U$. For any $\omega\in[0,1)$, we have
	\begin{align*}
		|\calR(\omega)-\hat\calR(\omega)|
		%&= \frac{ \Big| \|(I-P_{U}) \phi(\omega)\|  - \|(I-P_{\hat U}) \phi(\omega)\|  \Big|}{\|\phi(\omega)\|} \\
		%&\leq \frac{ \|(P_{U} - P_{\hat U}) \phi(\omega)\|  }{\|\phi(\omega)\|} \\
		&\leq \|P_{U} - P_{\hat U}\| 
		=\dist(U,\hat U), 
	\end{align*}
	where the last identity follows from \eqref{eqdistuhatuproj}. This shows that 
	$$
	\E \big(\|\calR -\hat\calR\|_\infty^2 \big)
	\leq \E \big( \dist(U,\hat U)^2\big) . 
	$$
	Using Theorem \ref{thm:Uperturb} completes the proof. 
	\end{proof}
	
\subsection{Preparation and proof of Theorem \ref{thmespritstability} part (a)}
	
	\label{proofthmespritstabilitya}
	
	The stability of ESPRIT relies on several deterministic perturbation results for the perturbation of $\widehat\Psi$ from $\Psi$, that can be estimated by matrix perturbation theory. 
	
	%{\color{blue} In literature, such perturbation bounds have been studied in \cite[Proposition 5.2]{nekrutkin2010perturbation}, while the  bounds in \cite{nekrutkin2010perturbation} can not be easily evaluated. Here we derive explicit upper bounds with explicit constants.} {\color{blue} Might not be good to say this... why do we need to give a comparison anyways? It seems like the reviewer just wants his paper cited.}

	\begin{lemma}
	\label{lemapsiper1}
	Suppose $M\ge S+1$ and  $\dist(\hat U,U)\leq  2^{-(S+2)}$. Then
	\begin{equation}
	\label{lemapsiper1eq0}
	\|\widehat\Psi - \Psi\|_2 
	\le {2}^{2S+4}   \dist(\hat U,U).
	\end{equation}
\end{lemma}

\begin{proof}
Before we proceed to the stability analysis, we need to point out an important feature of ESPRIT. ESPRIT is invariant to the specific choice of orthonormal basis for the column span of $\widehat U$. In other words, the eigenvalues of $\hat\Psi$ remain the same if one uses different orthonormal basis for the column space of $\widehat U$ \cite[Section III]{li2020super}. According to \cite[Lemma 2.1.3]{chen2020spectral}, we can properly choose an orthonormal basis for $\hat U$ such that
\begin{equation}
\|\hat U - U\|_2 \le \sqrt{2} \dist(\hat U,U).
\label{lemapsiper1eq1}
\end{equation}
According to \cite[Lemma 3]{li2020super}, $\sigma_S(U_0)$ is lower bounded such that
	\begin{equation}
		\min\big(\sigma_{S}(U_0),\sigma_{S}(U_1)\big)
		\geq 2^{-S} .
		\label{lemapsiper1eq4}
	\end{equation}
Combining \eqref{lemapsiper1eq1},  \eqref{lemapsiper1eq4} and the assumption of Lemma \ref{lemapsiper1} implies
	\[
	\begin{aligned}
	\|\hat U_0-U_0\|_2
	&\leq \|\hat U-U\|_2 
	\leq \sqrt{2} \dist(\hat U,U)
	%&\leq \frac{2^{-S}}{2}
	\leq \frac{1}{2} \sigma_{S}(U_0).
	\end{aligned}	\]
	This enables us to apply \cite[Lemma 6]{li2020super}, and obtain
	\begin{align*}
	\|\hat\Psi - \Psi \|_2  
	&\leq \frac{7\|\hat U - U\|_2}{\sigma_{S}^2(U_0)}
	\le \frac{7\sqrt{2}}{\sigma_{S}^2(U_0)}  \dist(\hat U,U),
	\end{align*}
	which further implies \eqref{lemapsiper1eq0} due to \eqref{lemapsiper1eq4}.
\end{proof}
		
 We next provide a deterministic bound between $\|\hat \Psi -\Psi\|_2$ and $\md(\hat \Omega,\Omega)$, with \cite[Lemma 2]{li2020super}. 
	
	\begin{lemma}
		\label{lemmamatching1}
		On the condition that $M \ge S+1$, the following statements hold:
		\begin{enumerate}[(a)]
			
			\item 
			\underline{Moderate perturbation regime.} We have
			\begin{equation*}
				{\rm md}(\hat\Omega,\Omega) \leq \frac{S^{3/2}\sqrt{M}}{\sigma_{S}(\Phi)}\|\hat\Psi-\Psi\|_2. 
			\end{equation*}
			
			\item 
			\underline{Small perturbation regime.} If additionally,
			\begin{equation*}
				\|\hat\Psi-\Psi\|_2
				\leq \frac{\sigma_{S}^2(\Phi) \Delta}{S^2 M},
			\end{equation*}
			then $
				{\rm md}(\hat\Omega,\Omega) 
				\leq 
				\|\hat\Psi-\Psi\|_2. 
            $
		\end{enumerate}	
		
	\end{lemma}

	{
	\begin{proof}[Proof of Theorem \ref{thmespritstability} part (a)]
		Assume for now that $\dist(\hat U,U)\leq 2^{-(S+2)}$. Combining Lemmas \ref{lemapsiper1} and \ref{lemmamatching1} part (a), we obtain the deterministic inequality,
		\begin{equation}
			\label{eq:esprithelp1}
			\md(\hat\Omega,\Omega)
			\leq \frac{4^{S+2}  S^{3/ 2}\sqrt{M}}{\sigma_{S}(\Phi)} \dist(\hat U,U). 
		\end{equation}
		If $\dist(\hat U,U)> 2^{-(S+2)}$, then the same inequality still holds since 
		\begin{align*}
		    \md(\hat\Omega,\Omega)
		\leq 1 
		&\leq \frac{2^{S+2}  \|\Phi\|_F }{\sigma_{S}(\Phi)} \dist(\hat U,U) \\
		&\leq \frac{4^{S+2}  S^{3/ 2}\sqrt{M}}{\sigma_{S}(\Phi)} \dist(\hat U,U),
		\end{align*}
		where we used that the Frobenius norm of $\Phi$ is $\sqrt {MS}$. 
		
		Hence, \eqref{eq:esprithelp1} holds regardless of the value of $\dist(\hat U,U)$. Squaring both sides of \eqref{eq:esprithelp1}, taking the expectation, and applying Theorem \ref{thm:Uperturb} completes the proof of part (a) of the theorem.
	\end{proof}

\subsection{Preparation and proof of Theorem \ref{thmespritstability} part (b)}
	
	\label{proofthmespritstabilityb}
	
	We begin by recalling two lemmas, and are essentially immediate consequences of \cite{cai2018rate}. 
	
	\begin{lemma}
		\label{lem:anru1}
		Suppose Assumption \ref{assump:main} holds. Define the events
		\begin{align*}
			\calE_1&:=\Big\{ \sigma_S^2(Y_L^* U) \geq  \frac{2}{3}\sigma_S^2(\Phi)\lambda_S(X)+\nu^2\Big\}, \\
			\calE_2&:=\Big\{\sigma_{S+1}^2( Y_L)\leq \frac{1}{3} \sigma_S^2(\Phi)\lambda_S(X) +\nu^2 \Big\}.		
		\end{align*}
		There exists a sufficiently large $D_\tau\geq 1$ depending only on $\tau$ such that the following hold. Suppose  
		\begin{equation}
		    \label{eq:noise1}
		    \xi:=\frac{L \sigma_S^2(\Phi)\lambda_S(X)}{M\nu^2}
		    \geq D_\tau.
		\end{equation}
		\begin{enumerate}
		    \item 
		There is a $c>0$ depending only on $\tau$ such that
		\begin{equation}
		\P \big( (\calE_1\cap \calE_2)^c\big)\lesssim_\tau \exp( -c M\xi).
		\label{lem:anru1:eq1}
		\end{equation}
		\item 
		There are $C,c>0$ depending only on $\tau$ such that for all $u\geq 0$,
		\begin{equation}
		    \label{eq:proj}
			\begin{aligned}
				&\P\Big( \|P_{Y_L^* U} ({Y_L}^* U_\perp)\|_2 \geq \frac{u\nu}{\sqrt L} \Big)\\
		        &\lesssim_\tau \exp\Big(CM-c\min\big(u^2,u\sqrt{M\xi}\big)\Big) 
		         +\exp\big( -c M\xi\big). 
			\end{aligned}
			%\label{lem:anru1:eq2}
		\end{equation}
		\end{enumerate}
	\end{lemma}
	
	\begin{proof}
	    The proof essentially follows from Lemma 4 in the supplement of \cite{cai2018rate}. For ease of notation, let $\eta:=\sqrt L/\nu$. To use the mentioned lemma, we first rescale by $\eta$ and take the conjugate transpose of $Y_L=\Phi X_L+E_L$ to get, 
    	\begin{equation}
    	    \label{eq:Yher}
    	\eta Y_L^*
    	=\eta (\Phi X_L)^* + \eta E_L^*,
    	\end{equation}
    	Notice that the $L\times M$ matrix $\eta E_L^*$ has entries that are i.d.d. $\mathcal{G}_{1,\tau}$ random variables. The the right singular vectors of $\eta Y_L^*$ are the left singular vectors of $Y_L$, and the conjugate transpose does not alter the singular values or their ordering. 
    	
    	In this proof, $C,c$ are constants that depend only on $\tau$, and their values may change from line to line. Now we apply Lemma 4 in the supplement of \cite{cai2018rate} to \eqref{eq:Yher} to get that, for all $u_1,u_2,u\geq 0$, 
    	\begin{equation}
    	    \label{eq:anru1help1}
    	    \begin{aligned}
    	        &\P\Big(\eta^2 \sigma_S^2(Y_L^* U)\leq (\eta^2\sigma_S^2(\Phi X_L)+L)(1-u_1)\Big) \\
    	        &\leq C \exp\Big( CS - c\big( \eta^2\sigma_S^2(\Phi X_L)+L\big)\min(u_1,u_1^2) \big) \Big), \\
    	        &\P\Big( \eta^2 \sigma_{S+1}^2(Y_L)\geq L(1+u_2)\Big) \\
    	        &\leq C\exp\Big(CM - cL \min(u_2,u_2^2) \Big). \\
				&\P\Big( \eta \|P_{Y_L^* U} ({Y_L}^* U_\perp)\|_2 \geq u \Big)\\
				&\lesssim_\tau \exp\Big(CM-c\min\Big(u^2,u\sqrt{\eta^2\sigma_S^2(\Phi X_L)+L}\Big)\Big) \\
				&\quad +\exp\big( -c (\eta^2 \sigma_S^2(\Phi X_L)+L)\big). 
			\end{aligned}
    	\end{equation}
        
        For the first bound in \eqref{eq:anru1help1}, we pick 
        $$
        u_1:=\frac{1}{3} \frac{\eta^2 \sigma_S^2(\Phi X_L) }{\eta^2\sigma_S^2(\Phi X_L)+L} \leq \frac{1}{3}. 
        $$
        Hence, $\min(u_1,u_1^2)=u_1^2$. The noise condition \eqref{eq:nu} implies that
        \begin{equation}
        \label{eq:nu2}
        \eta^2\sigma_S^2(\Phi X_L)+L
    	 \leq 2\eta^2\sigma_S^2(\Phi)\lambda_S(X), 
        \end{equation} 
        We also have $S\leq M$. So for sufficiently large $D_\tau$ depending only on $c,C$ (which only depend on $\tau$), condition \eqref{eq:noise1} and inequality \eqref{eq:nu2} allow us to say that  
    	  \begin{equation}
    	      \label{eq:anru1help2}
    	      \begin{aligned}
    	  &c \frac{\eta^4 \sigma_S^4(\Phi X_L) }{\eta^2\sigma_S^2(\Phi X_L)+L}-CS \\
    	  &\geq \frac{c}{2} \eta^2 \sigma_S^2(\Phi X_L)-CM 
    	  \geq \frac{c}{4} \eta^2 \sigma_S^2(\Phi)\lambda_S(X).
    	  \end{aligned}
    	  \end{equation}
    	Inserting this value of $u_1$ into \eqref{eq:anru1help1}, and using the above observations, we see that 
        \begin{equation*}
        \begin{aligned}
    	        \P(\calE_1^c)
    	        &\leq\P\Big(\eta^2 \sigma_S^2(Y_L^* U)\leq \frac{2}{3}\eta^2\sigma_S^2(\Phi X_L)+L\Big) \\
    	        &\leq C \exp\big( -c\eta^2 \sigma_S^2(\Phi)\lambda_S(X) \big). 
    	\end{aligned}
    	\end{equation*}
    	  
    	  For the second bound in \eqref{eq:anru1help1}, we pick 
    	  $$
    	  u_2=\frac{1}{3} \frac{\eta^2 \sigma_S^2(\Phi X_L)}{L}
    	  =\frac{1}{3} \frac{\sigma_S^2(\Phi X_L)}{\nu^2}.
    	  $$
    	  Notice that 
    	  \begin{align*}
    	  L\min(u_2,u_2^2) 
    	  &\geq \frac{1}{9} \min\Big( \eta^2 \sigma_S^2(\Phi X_L), \frac{\eta^4\sigma_S^4(\Phi X_L)}{L}\Big) \\
    	  &\geq \frac{1}{9} \frac{\eta^4\sigma_S^4(\Phi X_L)}{\eta^2 \sigma_S^2(\Phi X_L)+L}.
    	  \end{align*}
    	  Using the same argument as in \eqref{eq:anru1help2}, for $D_\tau$ sufficiently large, we see that 
    	  $$
    	  cL \min(u_2,u_2^2) -CM
    	  \geq \frac{c}{36} \eta^2 \sigma_S^2(\Phi)\lambda_S(X). 
    	  $$
    	  Combining these observations shows that 
    	  \begin{align*}
    	  \P(\calE_2^c)
    	  &\leq\P\Big( \eta^2 \sigma_{S+1}^2(Y_L)\geq L + \frac{1}{3} \eta^2\sigma_S^2(\Phi X_L) \Big) \\
    	  &\leq C\exp\big(-c\eta^2 \sigma_S^2(\Phi)\lambda_S(X) \big). 
    	  \end{align*}
        Combining the above with a union bound argument completes the proof of \eqref{lem:anru1:eq1}. 
        
    	  Finally, we proceed to simplify the probability bound for $\|P_{Y_L^* U} ({Y_L}^* U_\perp)\|_2$. Since
    	  $
    	  \sqrt{\eta^2\sigma_S^2(\Phi X_L)+L}
    	  \geq \sqrt{\eta^2\sigma_S^2(\Phi)\lambda_S(X)}
    	  $
    	  and
    	  $
    	  \eta^2\sigma_S^2(\Phi X_L)+L\geq \eta^2\sigma_S^2(\Phi)\lambda_S(X),
    	  $
    	  we can use \eqref{eq:anru1help1} to get \eqref{eq:proj}. 
\end{proof}

    The next lemma is a deterministic bound. It is Proposition 1 in \cite{cai2018rate} applied to equation \eqref{eq:Yher}. 
    
	\begin{lemma}
		\label{lem:anru2}
		Suppose Assumption \ref{assump:main} holds. If $\sigma_S(Y_L^* U) > \sigma_{S+1}(Y_L)$, then
		\begin{align*}
			d(\hat U, U)^2
			&\leq \frac{\sigma_S^2( Y_L^* U) \|P_{Y_L^* U} ({Y_L}^* U_\perp)\|_2^2	}{(\sigma_S^2( Y_L^* U)-\sigma_{S+1}^2( Y_L))^2}.
		\end{align*}
	\end{lemma}

	\begin{proof}[Proof of Theorem \ref{thmespritstability} part (b)]
		In this proof, we let $C>0, c>0$ be constants that only depend on $\tau$. Their values may change from one line to another. Let $\calE_1$ and $\calE_2$ be the events defined in Lemma \ref{lem:anru1} and set $\calE:=\calE_1\cap \calE_2$. 
		Define the event 
		$$
		\calH:= \Big\{ d(\hat U, U)
		\leq \frac{4^{S+2}\sigma_{S}^2(\Phi) \Delta}{S^2 M}  \Big\}
		=\Big\{ d(\hat U,U) \leq \sqrt 6 \rho\Big\}.
		$$
		In this proof, we consider the three events $\calE\cap \calH$, $\calE \cap \calH^c$, and $\calE^c$. The first event $\calE\cap \calH$ is a good event where we will be able to apply Lemma \ref{lemapsiper1}. The third event $\calE^c$ can be readily controlled using Lemma \ref{lem:anru1}. The majority of the proof pertains to the second event $\calE \cap \calH^c$, and to do this, we will estimate $\P(\calH^c|\calE)$. 
		
		Under event $\calE$, the inequalities in Lemmas \ref{lem:anru1} imply 
		\begin{align*}
		\sigma_S^2(Y_L^* U)-\sigma_{S+1}^2(Y_L)
		\geq \frac{1}{3} \sigma_S^2(\Phi) \lambda_S(X)
		>0. 
		\end{align*}
		This allows us to use Lemma \ref{lem:anru2} when $\calE$ occurs. Using the inequalities in Lemmas \ref{lem:anru1} again, $\sigma_S(Y_L^*U)\leq \sigma_S(Y_L)$ since $U$ has orthonormal columns, and $\nu^2\leq \sigma_S^2(\Phi)\lambda_S(X)$ from \eqref{eq:nu}, under event $\calE$,
	    \begin{align*}
	        d(\hat U, U)^2
			&\leq \frac{3(\sigma_S^2(\Phi)\lambda_S(X)+\nu^2) \|P_{Y_L^* U} ({Y_L}^* U_\perp)\|_2^2	}{\sigma_S^4(\Phi)\lambda_S^2(X)} \\
			&\leq \frac{6\|P_{Y_L^* U} ({Y_L}^* U_\perp)\|_2^2	}{\sigma_S^2(\Phi)\lambda_S(X)}. 
	    \end{align*}
		This inequality implies, 
		\begin{align*}
			\P(\calH^c|\calE)
			&\leq \P\Big( \|P_{ Y_L^* U} ({Y_L}^* U_\perp)\|_2 \geq \rho \sigma_S(\Phi) \sqrt{\lambda_S(X)} \ \big| \ \calE \Big) \\
			&= \P\Big( \|P_{ Y_L^* U} ({Y_L}^* U_\perp)\|_2 \geq \rho \sqrt{M \xi} \frac{\nu}{\sqrt L} \ \big| \ \calE \Big) \\
			&\leq \P\Big( \|P_{ Y_L^* U} ({Y_L}^* U_\perp)\|_2 \geq \rho \sqrt{M \xi} \frac{\nu}{\sqrt L} \Big) \frac{1}{\P(\calE)}.
		\end{align*} 
		The $1/\P(\calE)$ term can be safely ignored, since by Lemma \ref{lem:anru1} and \eqref{eq:SNRbound}, we have  
		$$
		\P \big( \calE^c\big)\lesssim_\tau \exp( -c M \xi)
		\leq \exp( -c D_\tau).
		$$
		For $D_\tau$ sufficiently large depending only on $\tau$ (since the implicit constants in this inequality only depend on $\tau$), we have $\P(\calE^c)\leq 1/2$ and $\P(\calE)\geq 1/2$.
		
		Let us continue with our goal of bounding $\P(\calH^c|\calE) $.  We use Lemma \ref{lem:anru2} to see that
		\begin{align*}
			\P(\calH^c|\calE)
			&\leq 2 \P\Big( \|P_{ Y_L^* U} ({Y_L}^* U_\perp)\|_2 \geq \rho \sqrt{M \xi} \frac{\nu}{\sqrt L} \Big) \\
			&\lesssim_\tau \exp\big(C  M-c \min\big(\rho^2 , \rho \big) M\xi\big) +\exp( -c M\xi ).
		\end{align*}
		Condition \eqref{eq:SNRbound} implies that $\xi \geq D_\tau/\min(1,\rho,\rho^2)$, so for sufficiently large $D_\tau$, we have
		$$
		\P(\calH^c|\calE) 
		\lesssim_\tau \exp\big(-c \min\big(1,\rho ,\rho ^2 \big) M\xi\big) 
		$$
		We use the inequality, $e^{-bu^2}\leq 1/u^2$ for all $u^2/\log (u^2) \geq 1/b$, with $u=\sqrt{\xi}$ and $b=c\min(1,a,a^2) M$, which is justified by \eqref{eq:SNRbound}, to finally see that
		\begin{equation}
			\label{eq:BGc}
			\P(\calH^c|\calE) 
			\lesssim_\tau \frac{M\nu^2}{\sigma_S^2(\Phi)\lambda_S(X)}.
		\end{equation}
		
		Now we are ready to finish the proof. Recall we have the trivial upper bound $\md(\hat\Omega,\Omega)\leq 1$. When event $\calE\cap \calH$ holds, we can employ Lemma \ref{lemapsiper1} and part (b) of Lemma \ref{lemmamatching1} to see that 
		$\md(\hat\Omega,\Omega)1_{\calE\cap \calH}\leq 4^{S+2}\dist(\hat U,U)1_{\calE\cap \calH}.$
		Then
		\begin{align*}
			&\E\big( \md(\hat\Omega,\Omega)^2\big) \\
			&= \E\big( \md(\hat\Omega,\Omega)^2 1_{\calE}\big)+\E\big( \md(\hat\Omega,\Omega)^2 1_{\calE^c}\big)\\
			&\leq \E\big( \md(\hat\Omega,\Omega)^2 1_{\calE\cap\calH}\big)+\E\big( \md(\hat\Omega,\Omega)^2 1_{\calE\cap \calH^c}\big) +\P(\calE^c) \\
			&\leq \E\big( 16^{S+2}\dist(\hat U,U)^2 1_{{\calE\cap\calH}}\big) + \P(\calE\cap \calH^c) +\P(\calE^c) \\
			&\leq 16^{S+2} \E\big(\dist(\hat U,U)^2\big) + \P(\calH^c|\calE) + \P(\calE\cap \calH^c).
		\end{align*}
		We use Theorem \ref{thm:Uperturb} to deal with the first term, and inequality \eqref{eq:BGc} for the second term. For the third term, by Lemma \ref{lem:anru1} and the inequality $e^{-bu^2}\leq 1/u^2$ with $u=\sqrt \xi$ and $b=cM$, we have
		$$
		\P(\calE^c)
		\lesssim_\tau \exp(-cM\xi)
		\lesssim_\tau  \frac{M\nu^2}{\sigma_S^2(\Phi)\lambda_S(X)}.
		$$
	\end{proof}
	}
	
	\subsection{Preparation and proof of Theorem \ref{thm:CR}}
\label{sec:CRproof}

The proof of the Cramer-Rao lower bound for clumps stated in Theorem \ref{thm:CR} requires several technical lemmas. While our approach is similar in spirit to that of \cite{lee1992cramer}, we need to extend the ideas in the referenced paper to several clumps. 

We introduce several functions of a single complex variable $z$. We typically use upper-case letters to denote functions of a single complex variable. The $m$-th Laurent coefficient of a $F\colon\C\to\C$ (expanded at $z=\infty$) is denoted by $b_m(F)$. Throughout, we fix a positive integer $n$ and $0<\epsilon<1$. We define the functions
\begin{equation}
	\begin{aligned}
		\label{eq:G}
		&G(z)
		:=\prod_{k=0}^{n-1} \frac{1}{z-k}
		=\sum_{m=n}^\infty \frac{b_m(G)}{z^m},\\ \quad\text{and}\quad
		&G_\epsilon(z)
		:=\prod_{k=0}^{n-1}   \frac{1}{z-\epsilon k}. 
		%	=\frac{1}{\epsilon^n}\prod_{k=0}^{n-1}  \frac{1}{z/\epsilon - k}
		%	=\sum_{m=n}^\infty \frac{b_m(G) \epsilon^{m-n}}{z^m}	    
	\end{aligned}
\end{equation}
For each $0\leq j < n$, let
\begin{equation}
	\begin{aligned}
		\label{eq:F}
		F_{\epsilon,j}(z)
		&:=\frac{G_\epsilon(z)}{z-j\epsilon}
		=\frac{1}{(z-j\epsilon)} \prod_{k=0}^{n-1} \frac{1}{z-\epsilon k}
		=\sum_{m=n+1}^\infty \frac{b_m(F_{\epsilon,j})}{z^m}.	  
	\end{aligned} 
\end{equation}
Notice that $F_{\epsilon,j}$ has simple poles at $k\epsilon$ for each $k\not=j$ and a pole of multiplicity 2 at $j\epsilon$. We have the partial fraction expansion of $F_{\epsilon,j}$,
\begin{equation}
	\label{eq:FCauchy}
	F_{\epsilon,j}(z) \prod_{k\not=j} (j\epsilon-k\epsilon) 
	=\frac{1}{(z-\epsilon j)^2} + \sum_{k=0}^{n-1} \frac{C_{\epsilon,j,k}}{z-k\epsilon}. 
\end{equation}

\begin{lemma}
	\label{lem:A}
	For each $\epsilon\in (0,1)$, positive integer $n$, and $0\leq j<n$, the coefficients $\{C_{\epsilon,j,k}\}_{k=0}^{n-1}$ defined in \eqref{eq:FCauchy} satisfy the system of $n$ equations:
	\begin{equation}
		\label{eq:system}
		-k(j\epsilon)^k = \sum_{\ell=0}^{n-1} (V_\epsilon)_{k,\ell} C_{\epsilon,j,\ell}, \quad 0\leq k< n, 
	\end{equation}
	where $V_\epsilon$ is a Vandermonde matrix with real nodes, 
	\[
	V_\epsilon
	:=
	\begin{bmatrix}
		1 & 1 &1 &\cdots & 1\\
		0 &\epsilon &2\epsilon & \cdots &(n-1)\epsilon \\
		\vdots &\vdots &\vdots  & &\vdots  \\
		0 &\epsilon^{n-1} &(2\epsilon)^{n-1} & \cdots &\big( (n-1)\epsilon \big)^{n-1}
	\end{bmatrix}. 
	\]
\end{lemma}

\begin{proof}
	Recall the following Laurent series expansions at $z=\infty$: for any $w\in\C$, we have
	\begin{align*}
		\frac{1}{z-w}
		&=\frac{1}{z(1-w/z)}
		=\sum_{m=1}^\infty \frac{w^{m-1}}{z^m}, \\
		\frac{1}{(z-w)^2}
		&=-\frac{d}{dz} \(\frac{1}{z-w}\)
		=\sum_{m=2}^\infty \frac{(m-1)w^{m-2}}{z^m}.
	\end{align*}
	Fix $0\leq j<n$. Using these expansions in \eqref{eq:F} and \eqref{eq:FCauchy}, we see that
	\begin{equation*}
		\begin{aligned}
			\label{eq:Flaurent}
			&\prod_{k\not=j} (j\epsilon-k\epsilon) \sum_{m=n+1}^\infty \frac{b_m(F_{\epsilon,j})}{z^m}\\
			&= \sum_{m=2}^\infty \frac{(m-1) (j\epsilon)^{m-2}}{z^m} + \sum_{k=0}^{n-1} C_{\epsilon,j,k} \sum_{m=1}^\infty \frac{(k\epsilon)^{m-1}}{z^{m}}. 
		\end{aligned}
	\end{equation*}
	Examining the $z^{-m}$ terms for $1\leq m\leq n$ in this equation, we see that
	\[
	0= \sum_{k=0}^{n-1} C_{\epsilon,j,k}, 
	\]
	and that for $2\leq m\leq n$, we have
	\begin{align*}
		0= (m-1) (j\epsilon)^{m-2} + \sum_{k=0}^{n-1} C_{\epsilon,j,k} (k\epsilon)^{m-1}.
	\end{align*}
	This is precisely the system of equations in \eqref{eq:system2}.
	
\end{proof}

The following lemma bounds the growth rate of the coefficients $\{C_{\epsilon,j,k}\}_{k=0}^{n-1}$ that were defined earlier.

\begin{lemma}
	\label{lem:C}
	For each $\epsilon\in (0,1)$, positive integer $n$, and $0\leq j<n$, the coefficients $\{C_{\epsilon,j,k}\}_{k=0}^{n-1}$ defined in \eqref{eq:FCauchy} satisfy
	\[
	\max_{0\leq k<n} |C_{\epsilon,j,k}|
	\leq 2^{n-1} (n-1)!.
	\]
\end{lemma}

\begin{proof}
	For $0\leq \ell<n$, let $\gamma_\ell$ denote the circle in the complex plane oriented counter-clockwise with center $\epsilon \ell$ and radius $\epsilon/2$. Integrating both sides of equation \eqref{eq:FCauchy} over $\gamma_\ell$, we see that 
	\[
	2 \pi i \, C_{\epsilon,j,\ell}
	=\( \int_{\gamma_\ell} F_{\epsilon,j}(z) \ dz\) \prod_{k\not=j} (j\epsilon-k\epsilon). 
	\]
	Here, we used that $(z-k\epsilon)^{-1}$ for each $k\not=\ell$ is holomorphic in an open disk containing $\gamma_\ell$ and that $\int_{\gamma_\ell} (z-\ell\epsilon)^{-2} \ dz=0$. Since $\gamma_k$ has circumference $\pi\epsilon$, the above imply
	\begin{align*}
		|C_{\epsilon,j,\ell}|
		&= \frac{1}{2\pi } \Big| \int_{\gamma_\ell} F_{\epsilon,j}(z) \ dz \Big| \ \prod_{k\not=j} |j\epsilon-k\epsilon| \\
		&\leq \frac{\epsilon^n (n-1)!}{2} \( \sup_{z\in\gamma_\ell} |F_{\epsilon,j}(z)|\). 
	\end{align*}
	Observe that for each $z\in \gamma_\ell$, we have $|z-\ell \epsilon| = {\epsilon}/{2}$ and 
	\[
	|z- k\epsilon| \geq \epsilon |\ell-k|-\frac{\epsilon}{2}\geq \frac{\epsilon}{2} \quad\text{if } k\not=\ell.
	\]
	Hence,
	$$
	\sup_{z\in\gamma_\ell} |F_{\epsilon,j}(z)|
	=\sup_{z\in\gamma_\ell} \(\frac{1}{|z-j\epsilon|} \prod_{k=0}^{n-1} \frac{1}{|z-\epsilon k|} \)
	\leq \frac{2^n}{\epsilon^n}.
	$$
	Combining the above inequalities completes the proof. 
\end{proof}

One of the key results is a following Taylor-like approximation of $\psi$ by linear combinations of $\phi$. It will be clear in the proof of Theorem \ref{thm:CR} why this is helpful.

\begin{lemma}
	\label{lem:taylorpsi}
	Fix $\xi\in\T$, positive integers $n$ and $M$, and $0\leq j<n$. For all sufficiently small $\epsilon>0$ depending only on $M$, $\xi$, and $j$, there exist coefficients $\{B_{\epsilon,j,k}\}_{k=0}^n\subset\R$ and $\{A_{\epsilon,j,\ell}\}_{\ell\geq n}\subset \C$ that do not depend on $\xi$ such that
	\[
	\psi(\xi+j\epsilon)
	=\sum_{k=0}^{n-1} B_{\epsilon,j,k} \, \phi(\xi+k\epsilon) +   \sum_{\ell=n}^\infty A_{\epsilon, j,\ell} \, \epsilon^{\ell-1} \frac{\phi^{(\ell)}(\xi) }{\ell!}.
	\]
	Moreover, we have that 
	\begin{equation}
		\label{eq:Abound}
		|A_{\epsilon,j,\ell}|
		\leq \ell j^{\ell-1} + n^{\ell+1} 2^{n-1} (n-1)!.
	\end{equation}
\end{lemma}

\begin{proof}
	We first fix $0\leq j<n$. Using that the collection of $M$ functions, $\xi\mapsto e^{2\pi ik\xi}$ for $0\leq k\leq M-1$, are analytic, for sufficiently small $\epsilon$ depending $M$, $\xi$ and $j$, we have 
	\[
	\phi(\xi+j\epsilon)
	=\sum_{\ell=0}^\infty \frac{\phi^{(\ell)}(\xi)}{\ell!} (j\epsilon)^\ell,
	\]
	where the above series converges absolutely. Differentiating each term, we see that
	\[
	\psi(\xi+j\epsilon)
	=\sum_{\ell=1}^\infty \frac{\phi^{(\ell)}(\xi) }{\ell!} \ell (j\epsilon)^{\ell-1}.
	\]
	
	We define $\{B_{\epsilon,j,k}\}_{k=0}^{n-1}$ as the unique solution to the system of equations: 
	\begin{equation}
		\label{eq:system2}
		k(k\epsilon)
		=\sum_{\ell=0}^{n-1} (V_\epsilon)_{k,\ell} B_{\epsilon,j,\ell}, \quad 0\leq k<n, 
	\end{equation}
	where $V_\epsilon$ is a Vandermonde matrix as defined in Lemma \ref{lem:A}. The $\{B_{\epsilon,j,k}\}_{k=0}^{n-1}$ are well-defined because $V_\epsilon$ is Vandermonde. Note that $B_{\epsilon,j,k}$ is independent of $\xi$. By definition, we have 
	\begin{align*}
		&\sum_{\ell=0}^{n-1} \frac{\phi^{(\ell)}(\xi)}{\ell!} \ell (j\epsilon)^{\ell-1} - \sum_{k=0}^{n-1} B_{\epsilon,j,k} \sum_{\ell=0}^{n-1} \frac{\phi^{(\ell)}(\xi)}{\ell!}  (k\epsilon)^\ell \\
		&\quad =\sum_{\ell=0}^{n-1} \frac{\phi^{(\ell)}(\xi)}{\ell!} \(\ell (j\epsilon)^{\ell-1} - \sum_{k=0}^{n-1} B_{\epsilon,j,k} (k\epsilon)^\ell\) =0.
	\end{align*}
	Using this equation, we see that
	\begin{align*}
		&\psi(\xi+j\epsilon) \\
		&=\sum_{\ell=0}^{n-1} \frac{\phi^{(\ell)}(\xi) }{\ell!} \ell(j\epsilon)^{\ell-1} + \sum_{\ell=n}^\infty \frac{\phi^{(\ell)}(\xi) }{\ell!} \ell(j\epsilon)^{\ell-1} \\
		&=\sum_{k=0}^{n-1} B_{\epsilon,j,k} \sum_{\ell=0}^{n-1} \frac{\phi^{(\ell)}(\xi)}{\ell!} (k\epsilon)^\ell + \sum_{\ell=n}^\infty \frac{\phi^{(\ell)}(\xi) }{\ell!} \ell(j\epsilon)^{\ell-1} \\
		&=\sum_{k=0}^{n-1}  B_{\epsilon,j,k} \, \phi(\xi+k\epsilon)\\
		&\quad +   \sum_{\ell=n}^\infty \frac{\phi^{(\ell)}(\xi) }{\ell!} \epsilon^{\ell-1}\( \ell j^{\ell-1}- \epsilon \sum_{k=0}^{n-1} B_{\epsilon,j,k} k^\ell \). 
	\end{align*}
	Defining the quantity
	\[
	A_{\epsilon,j,\ell}
	:=\ell j^{\ell-1}- \epsilon \sum_{k=0}^{n-1} B_{\epsilon,j,k} k^\ell
	\]
	proves the first statement of the lemma. 
	
	Observe that the system of equations \eqref{eq:system} and \eqref{eq:system2} are identical up to a negative sign, and so $B_{\epsilon,j,k}=-C_{\epsilon,j,k}$, where $C_{\epsilon,j,k}$ is defined in \eqref{eq:FCauchy}. Using the upper bound for $C_{\epsilon,j,k}$ given in Lemma \ref{lem:C} provides us with \eqref{eq:Abound}.  
%	\begin{align*}
%		|A_{\epsilon,j,\ell}|
%		%&\leq \ell j^{\ell-1} + \epsilon \sum_{k=0}^{n-1} |C_{\epsilon,j,k}| k^\ell \\
%		&< \ell n^{\ell-1} + 2^{n-1} (n-1)! \sum_{k=0}^{n-1} k^\ell \\
%		&\leq \ell n^{\ell-1} + 2^{n-1} (n-1)! \, n^{\ell+1}.
%	\end{align*}
\end{proof}

\begin{proof}[Proof of Theorem \ref{thm:CR}]
	Using that the maximum of a finite set exceeds its average and inequality \eqref{eq:cov1}, we see that
	\begin{gather}
		\label{eq:help0}
		\begin{split}
			&\E\big(\md(\hat\Omega,\Omega)^2 \big)
			\geq \frac{1}{S}\sum_{j=1}^S \E |\hat\omega_{\epsilon,j} - \omega_{\epsilon,j}|^2 \\
			&\quad\geq \frac{\nu^2}{2LS} \, \text{Tr}\( \big( \text{Re}\big( \Psi^* (I-P_{\Phi}) \Psi \odot X \big)\)^{-1}\).
		\end{split}
	\end{gather}
	Next, we use basic properties of the trace and spectral norm and that for any matrices $A$ and $B$, we have $\| \Re(A)\|_2\leq \|A\|_2$ and $\|A\odot B\|\leq \|A\|_2 \|B\|_2$, see \cite[Theorem 5.5.1]{horn1994topics}. Hence
	\begin{gather}
		\label{eq:help1}
		\begin{split}
			\text{Tr}&\( \big( \text{Re}\big( \Psi^* (I-P_{\Phi}) \Psi \odot X \big)\)^{-1}\)\\
			&\geq  \big\| \big( \text{Re} \big( \Psi^* (I-P_{\Phi}) \Psi \odot X \big)\big)^{-1} \big\|_2 \\
			&\geq \big\| \text{Re}\big( \Psi^* (I-P_{\Phi}) \Psi \odot X \big) \big\|^{-1}_2 \\
			&\geq \| X\|^{-1} \big\| \Psi^* (I-P_{\Phi}) \Psi \big\|^{-1}_2.
		\end{split}
	\end{gather}
	Since $(I-P_{\Phi})$ is a projection matrix, we have
	\begin{gather}
		\label{eq:help3}
		\begin{split}
			\big\| \Psi^* (I-P_{\Phi}) \Psi \big\|_2
			&= \big\| \big( (I-P_{\Phi}) \Psi \big)^* \big( (I-P_{\Phi}) \Psi \big) \big \|_2 \\
			&\leq \big\|(I-P_{\Phi})\Psi \big\|_F^2. 	
		\end{split}
	\end{gather}
	
	Each column of $\Psi$ has the form $\psi(\theta_r+j\epsilon)$. Let us fix any $1\leq r\leq R$ and $0\leq j<\lambda$. Applying Lemma \ref{lem:taylorpsi} (with $\xi=\theta_r$ and $n=\lambda>0$), for all sufficiently small $\epsilon$ depending on $M$, $\theta_r$ and $j$, 
	\begin{equation}
		\label{eq:help4}
		\begin{aligned}
			&(I-P_{\Phi})\psi(\theta_r+j\epsilon) 
			= \sum_{\ell=\lambda}^\infty A_{\epsilon,j,\ell} \, \epsilon^{\ell-1} \frac{(I-P_{\Phi}) \phi^{(\ell)}(\theta_r) }{\ell!},
		\end{aligned}
	\end{equation}
	where the above series converges absolutely and we used the crucial observation that $(I-P_{\Phi}) \phi(\theta_r+k\epsilon)=0$ for all $0\leq k<\lambda$. To make some the resulting expressions simpler, we set $M_0:=M-1$. A direct calculation shows that 
	\begin{equation}
		\label{eq:help5}			
        \begin{aligned}
			&\big\|(I-P_{\Phi}) \phi^{(\ell)}(\theta_r) \big\|_2
			\leq \big\| \phi^{(\ell)}(\theta_r) \big\|_2 \\
			&=\(\sum_{m=0}^{M_0} \Big|(-2\pi im)^\ell e^{-2\pi im\theta_r} \Big|^2\)^{1/2}
			\leq (2\pi)^\ell M_0^{\ell+1/2}. 		    
		\end{aligned}
	\end{equation}
	It follows from equation \eqref{eq:help4}, inequality \eqref{eq:help5}, and the upper bound for $|A_{\epsilon,j,\ell}|$ given in Lemma \ref{lem:taylorpsi}, that
	\begin{equation}
		\label{eq:help6}
		\begin{aligned}				
		&\big\|(I-P_{\Phi})\psi(\theta_r+j\epsilon)\big\|_2 \\
			&\quad\leq \sum_{\ell=\lambda}^\infty |A_{\epsilon,j,\ell} | \, \epsilon^{\ell-1} \frac{(2\pi)^\ell  M_0^{\ell+1/2} }{\ell!} \\
			&\quad\leq M_0^{3/2} (\epsilon M_0)^{\lambda-1} \sum_{\ell=\lambda}^\infty |A_{\epsilon,j,\ell} | \frac{(2\pi)^\ell (\epsilon M_0)^{\ell-\lambda} }{\ell!} \\
			&\quad\leq M_0^{3/2} \big(\epsilon M_0\big)^{\lambda-1} \sum_{\ell=\lambda}^\infty C_{\ell,\lambda} \frac{(2\pi)^\ell (\epsilon M_0)^{\ell-\lambda} }{\ell!},
		\end{aligned}
	\end{equation}
	where we have defined
	$
	C_{\ell,\lambda}:= \ell \lambda^{\ell-1} + \lambda^{\ell+1} 2^{\lambda-1} (\lambda-1)!. 
	$
	By making $\epsilon$ even smaller if necessary, the series on the right hand side of \eqref{eq:help6} converges to some value that can be upper bounded by a $C>0$ that depends only on $\lambda$.	Note that from this equation, a necessary condition is that $\epsilon< (2\pi \lambda M_0)^{-1}$. Recall that $\epsilon$ also depends on $\theta_r$ and $j$. Making $\epsilon$ even smaller if necessary so that \eqref{eq:help6} holds for all $\theta_r$ and $j<\lambda$, we see that
	\begin{equation}
		\begin{aligned}
			\label{eq:help7}
			\big\|(I-P_{\Phi})\Psi \big\|_F^2 
			&=\sum_{r=1}^R \sum_{j=0}^{\lambda-1} \big\|(I-P_{\Phi})\psi(\theta_r+j\epsilon)\big\|^2_2 \\
			&\leq C^2 R\lambda M_0^3 \big(\epsilon M_0\big)^{2\lambda-2}. 
		\end{aligned}
	\end{equation}
	Combining inequalities \eqref{eq:help0}, \eqref{eq:help1}, \eqref{eq:help3}, and  \eqref{eq:help7} completes the proof.
\end{proof}

	\bibliographystyle{plain}
	\bibliography{MSSreference,SRlimitFourierbib}

\end{document}